\documentclass[10pt,journal,compsoc]{IEEEtran}

\usepackage[linesnumbered,ruled,vlined]{algorithm2e}%[ruled,vlined]{
\usepackage{algpseudocode}
\usepackage{amsmath}
\usepackage{amsthm}
\usepackage{color}
\usepackage{amssymb}
\usepackage{graphicx}
\usepackage{url}
\usepackage{hyperref}
\usepackage{enumerate}
\usepackage{subfigure}
\usepackage{cases}

\usepackage{todonotes}
\setuptodonotes{inline}

\newtheorem{lemma}{Lemma}[section]
\newtheorem{remark}{Remark}[section]
\newtheorem{theorem}{Theorem}[section]

\newtheorem{definition}{Definition}[section]

\newcommand{\biggg}{\bBigg@{3}}

\newcommand{\Biggg}{\bBigg@{3.5}}

\newcommand{\bigggg}{\bBigg@{4}}

\newcommand{\Bigggg}{\bBigg@{4.5}}

\begin{document}
\title{A Distributed Privacy-Preserving Learning Dynamics in General Social Networks}

\author{Youming~Tao,~\IEEEmembership{Student Member,~IEEE,}
        Shuzhen~Chen,~\IEEEmembership{Student Member,~IEEE,}
        Feng~Li, ~\IEEEmembership{Member,~IEEE,}
        Dongxiao~Yu, ~\IEEEmembership{Senior Member,~IEEE,}
        Jiguo~Yu, ~\IEEEmembership{Fellow,~IEEE,}
        Hao~Sheng,~\IEEEmembership{Member,~IEEE,}% <-this % stops a space
\IEEEcompsocitemizethanks{\IEEEcompsocthanksitem Y. Tao, S. Chen, F. Li, and D. Yu are with School of Computer Science and Technology, Shandong University, Qingdao, 266237, China. E-mail: \{youming.tao, szchen\}@mail.sdu.edu.cn, \{fli, dxyu, xzcheng\}@sdu.edu.cn. %\protect\\
\IEEEcompsocthanksitem J. Yu is with Big Data Institute, Qilu University of Technology, Jinan, 250353, China. Email: jiguoyu@sina.com.%\protect\\
\IEEEcompsocthanksitem H. Sheng is with State Key Laboratory of Software
Development Environment, School of Computer Science and Engineering, Beihang University, Beijing 100191, China, and Beijing Advanced Innovation Center for Big Data and Brain Computing, Beihang University, Beijing, 100191, P.R.China. E-mail: shenghao@buaa.edu.cn.

%\IEEEcompsocthanksitem J. Doe and J. Doe are with Anonymous University.
}% <-this % stops an unwanted space
}

% \markboth{Journal of \LaTeX\ Class Files,~Vol.~14, No.~8, August~2015}%
% {Shell \MakeLowercase{\textit{et al.}}: Bare Demo of IEEEtran.cls for Computer Society Journals}

\IEEEtitleabstractindextext{%
\begin{abstract}
  In this paper, we study a distributed privacy-preserving learning problem in social networks with general topology. The agents can communicate with each other over the network, which may result in privacy disclosure, since the trustworthiness of the agents cannot be guaranteed. Given a set of options which yield unknown stochastic rewards, each agent is required to learn the best one, aiming at maximizing the resulting expected average cumulative reward. To serve the above goal, we propose a four-staged distributed algorithm which efficiently exploits the collaboration among the agents while preserving the local privacy for each of them.  In particular, our algorithm proceeds iteratively, and in every round, each agent i) randomly perturbs its adoption for the privacy-preserving purpose, ii) disseminates the perturbed adoption over the social network in a nearly uniform manner through random walking, iii) selects an option by referring to the perturbed suggestions received from its peers, and iv) decides whether or not to adopt the selected option as preference according to its latest reward feedback. Through solid theoretical analysis, we quantify the trade-off among the number of agents (or communication overhead), privacy preserving and learning utility. We also perform extensive simulations to verify the efficacy of our proposed social learning algorithm.

\end{abstract}
}

\maketitle

\IEEEdisplaynontitleabstractindextext

% \fi
\IEEEpeerreviewmaketitle

\IEEEraisesectionheading{\section{Introduction}\label{sec:intro}}
  Given a set of options which yield unknown stochastic rewards/payoffs, learning the best among them is a commonly encountered issue in a wide spectrum including human society \cite{ShenWJZ-IJCAI15}, robotics \cite{PiniBFDB-ICSI12} and biology \cite{SeeleyB-BES99}. This problem can be casted as the following sequential decision-making problem: every individual (a.k.a., agent) sequentially selects one of the unknown option to observe its reward feedback and updates its adoptions (i.e., its preference to the options) accordingly; the goal is to maximize the expected cumulative reward yielded in the above learning process without the prior knowledge about the options' stochastic qualities.

  In a social group, each agent can share its experience with each other, to improve the efficiency of the above learning process. As shown in \cite{CelisKV-PODC17}, an iterative social learning approach to the above problem consists of the following two stages in each step: every agent first takes an option sample according to the options' popularities among all the agents, and then decides whether or not to adopt the sampled option as preference based on the latest observation on its stochastic reward signal. Such a ``sampling-and-adopting'' social learning paradigm does not need any historical observations and thus can work with limited local memory at each agent \cite{SuZL-SIGMETRIC19}; nevertheless, it entails global information (i.e., all agents' latest adoptions) as input to calculate the options' popularities. Unfortunately, a well structured communication network over the agents (e.g., a complete graph such that each agent can easily calculate the populairties of the options through one-hop communications) may not always be available, while collecting the global information over a social network with general topology usually results in considerable communication overhead. Therefore, it is very challenging to enable efficient collaboration among the agents over a general social network.

  Our another concern is the privacy issue in the collaboration among the agents, since the agents in a social network usually are not forced to behave trustfully to each other and sharing experience with other untrusted peers may lead to privacy disclosure for each agent. \textit{Local Differential Privacy} (LDP) is a privacy-preserving mechanism where data owners perturb their private data locally before sharing them \cite{Kasiviswanathan-JOC11,DuchiJW-FOCS13}. The concept of LDP has been applied in distributed learning frameworks. The data owners add noise to their local gradients \cite{ShokriS-CCS15,AbadiCGMMTZ-CCS16} or local model parameters \cite{WeiLDMYFQP-TIFS20} before reporting them to a central aggregation server. In addition, LDP has also been applied in on-line decision-making problems such that the rewards of options are perturbed before being reported to a central decision maker \cite{GajaneUK-ALT18,WangZWCKW-RecSys20,RenZLS-arXiv20}. Although the potential of LDP has been recognized, we consider a very different decentralized social learning process where untrustworthy agents communicate private experience with each other over a network without central infrastructures, so as to make decisions collaboratively. Hence, how to apply LDP in our decentralized learning process is still an open problem.

  In this paper, we propose a privacy-preserving distributed social learning algorithm for general social networks. It proceeds iteratively and includes the following four stages in each round:
  \begin{itemize}
    \item \textbf{Perturbing}: By leveraging the notion of LDP, each agent applies a randomized perturbation to its current adoption for the purpose of privacy preserving.
    \item \textbf{Disseminating}: We propose a random walk-based information dissemination method, by which each agent disseminates its perturbed adoption over the (multi-hop) social network with general topology.
    \item \textbf{Sampling}: Each agent then selects one of the options according to the perturbed suggestions received from its peers in the network.
    \item \textbf{Adopting}: Each agent finally decides whether or not to adopt the option selected in the last stage, according to its recent stochastic quality signal.
  \end{itemize}
  The above algorithm inherits the efficiency of the state-of-the-art ``sampling-and-adopting'' social learning paradigm, such that each agent with limited local memory maintains its current adoption (or preference) only. Furthermore, it integrates a randomized perturbing mechanism and a random walk-based information dissemination method for privacy-aware learning in general social networks. Nevertheless, according to our brief sketch on our proposed algorithm (especially the first two stages), on on hand, each agent shares its adoptions in a randomized manner such that the knowledge received by an agent is incomplete; on the other hand, the knowledge is perturbed and its usability in the learning process may be reduced by the perturbation. Therefore, some fundamental questions are still open: \textit{With the randomly sampled and perturbed knowledge for each agent, how does our four-staged social learning algorithm converge to the off-line optimal solution? Specifically, how is the expected average cumulative reward yielded by the adoption policy maximized? To what extend the privacy preservation can be ensured and at what cost?} In this paper, we answer the above questions by solid theoretical analysis. We also perform extensive simulations to verify the efficacy of our algorithm.

  The remainder of this paper is organized as follows. We survey related literature in Sec.~\ref{sec:rel}. In Sec.~\ref{sec:model}, we discuss about the motivations behind our algorithm design, before introducing our system model and formulating our learning problem. Therein, we also introduce some preliminaries which will be useful to our algorithm design and analysis. We then present the design of our four-staged privacy-preserving distributed social learning algorithm and the corresponding analysis in Sec.~\ref{sec:alg} and Sec.~\ref{sec:analysis}, respectively. The simulation results are reported in Sec.~\ref{sec:exp}. We finally conclude our paper and propose some promising research directions for future in Sec.~\ref{sec:con}.

\vspace{-2ex}
\section{Related Work}\label{sec:rel}
  \subsection{Social Learning}
  %
    %Learning the best option for an isolated individual has been proved impossible in \cite{54,53}, but the successful learning is often observed in social groups. This may because that by listening to experiences from the past, society can take advantage of others' experience through observing their neighbors and avoid repeating the mistakes of the past\cite{55,56}. Hence, distributed social learning dynamics has attracted significant amount of attention in the past few decades. 

    Optimizing the decision-making process to maximize the expected reward is an essential problem for socialized individuals. Learning the best option for an isolated individual with time-invariant finite memory has been proved impossible in \cite{CoverH-TIT70,XuY-SIGMETRICS18}. Nevertheless, a common wisdom may suggest that interacting with each others to share the choices in the social group may contribute to the success of the learning process. Specifically, one can learn the experiences from its peers so as to avoid making similar mistakes in the decision-making process \cite{Bandura-book69,ImmorlicaMT-WWW19}. Whereas the learning algorithms in early studies involve either sampling stage~\cite{BoydR-AJOS85} or adopting stage~\cite{Henrich-AA10}, it is demonstrated in \cite{McelreathBELW-PT08} that combining the two steps together could be a better learning strategy empirically. The two-staged learning strategy is then applied in sociology and economics \cite{EllisonF-QJE95,Cabrales-IER00,KrafftZPPASTP-CoRR16}, where the two stages are both crucial for the learning process. Our work is partially inspired by the recent work \cite{CelisKV-PODC17}. It investigates the dynamics of the two-staged learning strategy which entails global information as input, such that one in the social network can be aware of the choices of all the others. Unfortunately, single-hop social networks are not always available, whereas collecting the choices over a general (multi-hop) social network may induce considerable communication overhead. Hence, in this paper, one of our contributions is to study the learning dynamics in general social networks where the individuals interact with each other by multi-hop communications.

  \vspace{-2ex}
  \subsection{Privacy Preservation in Machine Learning} \label{ssec:prilearning}
    Privacy has emerged as one of the main concerns in machine learning research \cite{LiuDSRFL-ACS21}. One choice is to apply cryptography-based methods, e.g., secure multi-party computation \cite{MohasselZ-SP17,BonawitzIKMMPRSS-CCS17} and homomorphic encryption \cite{DannerJ-DAIS15,WangDCCZCH-TKDE18}. However, the cryptography-based methods may induce considerable computation overhead. Therefore, another branch of studies rely on the notion of \textit{Differential Privacy} (DP) \cite{DworkR-FTTCS14}. In the traditional \textit{global} DP, data are first collected from their owners to a trusted third-party. When the data are queried, the aggregated result for the query is perturbed by the third-party before being released to the untrusted requesters \cite{GuptaHRU-JOC13,CaoYXX-TKDE19}. The application of the global DP in federated learning is investigated in \cite{NaseriHC-NDSS22}. In each iteration of FL, participants first submit their local models to a trusted central server, and the server then feeds a perturbed aggregation to the participants. The noise-adding (or perturbing) technique is also used in designing defense mechanisms against attacks to well trained prediction models. For example, \cite{JiaSBZG-CCS19} proposed to add noise to the prediction results of a target classifier so as to defend against black-box membership inference attacks. In \cite{OrekondySF-ICLR20}, predictions are perturbed to poison the training objectives of model stealing attackers.

    According to the definition of the global DP, it is used only when there is a trusted third-party to collect data from individual data owners and to protect their privacy in the meanwhile; however, the trustworthiness of the (central) third-party infrastructure may not be guaranteed, and we have to consider locally protecting the privacy of the data owners. To serve the above goal, \textit{local} DP (LDP) is proposed as a distributed variant of DP. For example, in \cite{ShokriS-CCS15,AbadiCGMMTZ-CCS16}, data owners add noise to (or perturb) the gradients calculated locally before reporting them to a central server where the perturbed gradients are aggregated to update the model parameters. In \cite{WeiLDMYFQP-TIFS20}, the data owners add noise to the locally calculated model parameters and then upload them to an aggregation server. In \cite{NaseriHC-NDSS22,WeiLDMSZP-TMC2022}, LDP is applied to achieve user-level privacy preservation in FL, by letting users upload noised local models. In addition, LDP has also been applied in on-line decision-making problems where decisions on selecting among a group of unknown options are made according to the reward feedbacks of the options in an on-line manner. For example, in  \cite{GajaneUK-ALT18,WangZWCKW-RecSys20,RenZLS-arXiv20}, a central decision maker makes selection decisions according to the perturbed observations on the rewards of the options. The above proposals all consider the privacy issue in a centralized learning process where multiple data owners need to report their private information (e.g., local gradients, model parameters or reward observations) to an untrusted central infrastructure. In contrast, we are interested in investigate the privacy issue in a decentralized learning process where multiple decision makers (i.e., agents) do not trust each other but have to share their private experience with each other for making decisions collaboratively.

\vspace{-2ex}
\section{Motivations, Models and Preliminaries}\label{sec:model}
  In this section, we first present concrete examples to motivate our algorithm in Sec.~\ref{ssec:mot}. We then introduce our system model in Sec.~\ref{ssec:sys}, and formulate our problem of social learning dynamics in Sec.~\ref{ssec:goal}. We also introduce some preliminaries in Sec.~\ref{ssec:prel}. Frequently used notations throughout this paper are summarized in \textbf{Table}~\ref{tab:notation}.
  \begin{table*}
  \caption{Frequently used symbols and notations.} \label{tab:notation}
  \vspace{-3ex}
	\begin{tabular}[t]{|p{3.5cm}|p{13.5cm}|}
	\hline
	  $\mathcal G=(\mathcal N, \mathcal E)$ & A social graph consisting of agents $\mathcal N$ and edges $\mathcal E$ \\ \hline
    %
	  %$\mathcal N_i \subseteq \mathcal N$ & The set of the neighbors of agent $i$ \\ \hline
    %
	  $N_i=|\mathcal N_i|$ & $\mathcal N_i \subseteq \mathcal N$ dentoes the set of the neighbors of agent $i$, while $N_i$ is the number of $i$'s neighbors \\ \hline
    $\mathcal M$ & A set of $M=|\mathcal M|$ unknown options \\ \hline
    %
	  %$M=|\mathcal M|$ & The number of options \\ \hline
    %
	  $\Phi^r_j \in\{0,1\}$  & Random quality indicator of option $j$ in round $r$  \\ \hline
	  $\eta_j\in[0,1]$  & Probability of $\Phi^r_j = 1$   \\ \hline
	  $\mathbf{X}^r_i = [X^r_{i,1}, \cdots, X^r_{i,M}]$ & Adoption vector of agent $i$ in round $r$ where $X^r_{i,j} \in \{0,1\}$ indicates if agent $i$ adopts option $j$ in round $r$   \\ \hline
	  $\mathbf{\widetilde X}^r_i = [\widetilde X^r_{i,1}, \cdots, \widetilde X^r_{i,M}]$ & The perturbed adoption vector of agent $i$ in round $r$ where $\widetilde X^r_{i,j}$ is the random perturbation of $X^{r-1}_{i,j}$  \\ \hline
	  $\mathcal V^r$ & The set of the (replicated) perturbed adoption vectors disseminated in round $r$\\ \hline
    $\mathcal V^r_i \subseteq \mathcal V^r$ & The set of the (replicated) perturbed adoption vectors sampled by agent $i$ in round $r$\\ \hline
	  %
    %$V^r = |\mathcal V^r|$ & The number of the (replicated) perturbed adoption vectors disseminated in round $r$.\\ \hline
    %
	  %$V^r_{i} = |\mathcal V^r_i|$ & The number of the (replicated) perturbed adoption vectors sampled by agent $i$ in round $r$.\\ \hline
	  %
	  %$\mathbf{X}^r_{i}$ & Adoption vector of agent $i$ in round $r$\\ \hline
	  %
	  %$\widetilde{\mathbf{X}}^r_{i}$ & Random perturbation of $\widetilde X^r_{i,j}$\\ \hline
	  %
	  $Q^r_j$  & The popularity of option $j$ in round $r$.   \\ \hline
	  $\widetilde Q^r_{i,j}$ (resp. $\widehat Q^r_{i,j}$)  & Agent $i$'s unnormalized (resp. normalized) estimate on $Q^{r-1}_j$ in round $r$.   \\ \hline
    %
	  %$\widehat Q^r_{i,j}$  & Agent $i$'s normalized estimate on $Q^{r-1}_j$ in round $r$ \\ \hline
	  %
    %
    $\mathcal S^r_j = \{i\in\mathcal N \mid Y^r_{i,j} = 1\}$ & The set of agents selecting option $j$ in the sampling stage of round $r$ \\ \hline
    $S^r_j = |\mathcal S^r_j|$ & The number of agents selecting option $j$ in the sampling stage of round $r$ \\ \hline
    $D^r_j = \sum^N_{i=1}X^r_{i,j}$ & The number of agents adopting option $j$ in round $r$ \\ \hline
    $D^r = \sum^M_{j=1}D^r_{j}$ & The number of agents with non-null adoption vectors in round $r$ \\ \hline
    \end{tabular}
  \end{table*}
  
  %\vspace{-2ex}
  \subsection{Motivations} \label{ssec:mot}
    As mentioned in Sec.~\ref{sec:intro}, selecting among a set of unknown option is a very common issue in daily life, while our aim is to utilize the collaboration among the individual decision makers (i.e., agents) in a social group and to guarantee the privacy of the individuals when they collaborate with each other. This idea actually is motivated by many real-world applications. For example, 
    \begin{itemize}
      \item \textbf{Clinical trials}. Suppose there is an illness with multiple treatments for patients (a.k.a., options) and a group of experimenters (a.k.a., agents) sequentially choose among the given treatments. The goal for the experimenters is to maximize the number of cured patients without prior knowledge about the effects of the treatments. On one hand, the experimenters can collaborate with each other by exchanging their experience; on the other hand, they are not willing to publish their current treatment in use, as doing so may induce the leakage of sensitive information such as the patients' health conditions and genomes.
      \item \textbf{Procurement of financial products}. In many economic scenarios, individuals (a.k.a., agent) need to make a sequence of decisions on selecting among different financial products (a.k.a., options). The individuals can share their latest adopted options with each other for common prosperity; nevertheless, an individual would not like to let others know its actual selection on the financial products, especially the trustworthiness of the other peers cannot be ensured.
    \end{itemize}

    The privacy-preserving collaboration is also demanded in many other applications such as advertising, recommendation systems etc. Inspired by these application examples, we illustrate a social learning system in Fig.~\ref{fig:sys}, where multiple agents in a social network can collaborate with each other by exchanging private experience through the communication links. Nevertheless, since the trustworthiness of the agents cannot be ensured, each of them has to preserve its privacy when collaborating with its peers. Note that we do not assume specific untrusted agents in this paper and we suppose each agent trusts none of the others. This system model will be formalized in the following.
    \begin{figure}[htb!]
      \centering
      \includegraphics[width=0.9\linewidth]{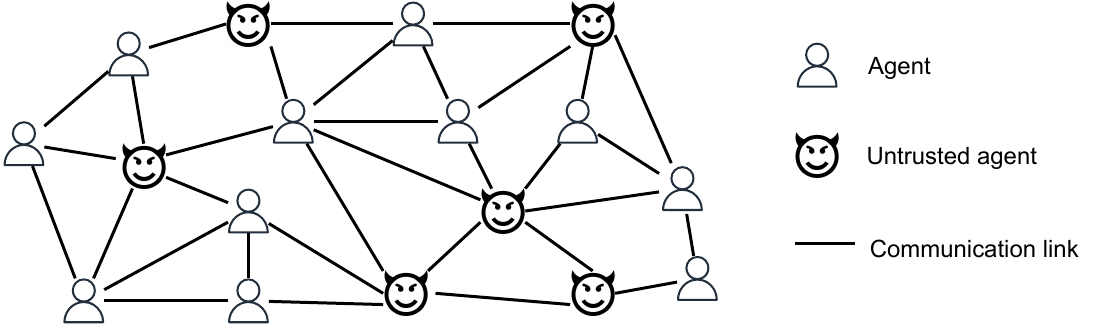}
      \vspace{-2ex}
      \caption{The agents in the graph can collaborate with each other by exchanging private information through a graph with arbitrary topology, but their trustworthiness cannot be guaranteed.}
      \label{fig:sys}
    \end{figure}
    %
    
  %\vspace{-6ex}
  \subsection{System Model} \label{ssec:sys}
    We consider a social network represented by an undirected graph $\mathcal G = (\mathcal N, \mathcal E)$. $\mathcal N = \{1,2,\cdots,N\}$ denotes a set of $N$ agents while $\mathcal E$ denotes the set of the edges between the agents. For $\forall i,i' \in \mathcal N$ and $i \neq i'$, we have an edge $(i,i')\in \mathcal E$ if they can exchange messages with each other. Let $\mathcal N_i = \{i'\in\mathcal N \mid (i', i)\in\mathcal E\}$ denote the neighbors of agent $i$ and $N_i$ be the size of $\mathcal N_i$ (i.e., the degree of agent $i$ in graph $\mathcal G$). Assume each agent $i$ initially is aware of $N_i$. Without loss of generality, we suppose $\mathcal G$ is a \textit{connected} and \textit{non-bipartite} graph. We assume that the network is well synchronized such that time can be divided into a sequence of rounds $r = 1, 2, \cdots, R$, each of which consists of $\Delta$ unit time slots. We adopt a relaxed CONGEST communication model, such that each agent is allowed to send only $\mathcal O(g(N))$ messages of $\mathcal O(\log N)$ bits over each edge in a slot \footnote{The formal definition of $g(N)$ will be given in Sec.~\ref{sssec:facts}. Informally, $g(N)$ is such that $\ln N<g(N)<N$ when $N$ is sufficiently large.} We also suppose each agent has a privacy budget $\varepsilon$.

    Suppose there are $M$ options $\mathcal M = \{1,\cdots, M\}$. Each option $j \in \mathcal M$ is associated with a random \textit{quality} indicator $\Phi^r_j\in\{0,1\}$ in each round $r$ \footnote{As will be shown later, our algorithm is readily compatible to quality indicators varying randomly across slots.}. Specifically, we have $\Phi^r_j=1$ if option $j$ is ``good'' in round $r$ such that the agents choosing it can gain reward; otherwise, $\Phi^r_j=0$. For $\forall j \in \mathcal M$, $\Phi^1_j, \Phi^2_j, \cdots, \Phi^R_j$ are drawn independently and identically from an \emph{unknown} Bernoulli distribution parameterized by $\eta_j$, i.e., $\Phi^r_j \sim \mathsf{Bernoulli}(\eta_j)$, such that $\mathbb P[\Phi_j^r=1]=\eta_j$ and $\mathbb P[\Phi_j^r=0]=1-\eta_j$. Without loss of generality, suppose $\eta_1>\eta_2\ge\cdots\ge\eta_M$, such that the first option is the best.

    During the learning process, the agents need to communicate with their neighbors for exchanging private information (i.e., their latest adoptions in our case), which entails a high demand on privacy preserving. For example, for each agent, if one of its neighbors is subverted by an adversary, the adversary may be able to eavesdrop on the private information shared by the agent. Therefore, in this paper, we consider a strong threat model by leveraging local privacy \cite{Kasiviswanathan-JOC11,DuchiJW-FOCS13}, assuming that each agent trusts none of the others (especially its neighbors).

  \subsection{Social Learning Dynamics} \label{ssec:goal}
    As introduced in Sec.~\ref{sec:intro}, the social learning algorithm proceeds iteratively in rounds $r=1,2,\cdots$. Let $X_{i,j}^r \in \{0,1\}$ be a binary variable indicating if agent $i$ adopts option $j$ in round $r$. We assume that each agent $i$ adopts at most one option in each round such that $\sum^M_{j=1} X^r_{i,j} \leq 1$ for $\forall r$. Without prior knowledge on $\eta_j$ for $\forall j\in\mathcal M$, the learning goal is to minimize the following regret function
    \begin{equation} \label{eq:reg}
      \mathsf{Regret}_N(R)=\eta_1-\frac{1}{R}\sum\limits_{r=1}^{R}\sum\limits_{j=1}^{M}\mathbb{E}\left[ Q_j^{r-1}\Phi_j^r\right],
    \end{equation}
    where $Q_j^r = \frac{\sum_{i=1}^{N}X_{i,j}^r}{\sum_{j'=1}^{M}\sum_{i=1}^{N}X_{i,j'}^r}$
    %
    % \begin{equation} \label{eq:frac}
    %   Q_j^r = \frac{\sum_{i=1}^{N}X_{i,j}^r}{\sum_{j'=1}^{M}\sum_{i=1}^{N}X_{i,j'}^r}
    % \end{equation}
    %
    denotes the \emph{popularity} of option $j$ in round $r$, namely the fraction of agents adopting option $j$ in round $r$. Initially, we assume $Q_j^0=\frac{1}{M}$ for $\forall j\in\mathcal M$ \footnote{Similar to \cite{CelisKV-PODC17}, such an assumption of equal popularities is not crucial to our results. Our results hold with arbitrary initial conditions.}.
    The regret function measures the difference between the off-line optimal policy and our on-line learning algorithm in terms of expected cumulative reward averaged over $N$ agents in $R$ rounds. In the off-line optimal policy, $\eta_1, \cdots, \eta_M$ are known for each agent and the agent can always adopt the best option as preference; while in our algorithm, each agent adopts one of the options sequentially with no prior knowledge about $\eta_j$. In fact, the regret function reflects how the utility of our learning algorithm (represented by the expected average cumulative reward yielded by the adoption policy learnt by our algorithm) approaches the optimum. Smaller regret implies each agent learns the optimal option more efficiently through our algorithm, resulting in higher expected average cumulative reward and thus higher learning utility.

    In this paper, we investigate a distributed social learning algorithm for general (multi-hop) social networks, such that the agents work collaboratively by exchanging experience with each other through a network with general topology, so as to minimize the  regret function. Moreover, our another concern is the privacy issue for the agents such that their communications will not result in privacy disclosure.

  \vspace{-2ex}
  \subsection{Preliminaries} \label{ssec:prel}
    \subsubsection{Metropolis-Hasting Random Walk} \label{sssec:rw}
      In this paper, we leverage the notion of random walk such that each agent samples the distribution of the adoptions in the social network. In each step, a walk carries an information token from the current agent to a random neighbor or itself. Particularly, for a \textit{Metropolis-Hasting Random Walk} (MHRW), in each step, $i$ forwards the token to a randomly chosen neighbor (or itself) $i'$ according to probability $\Psi(i,i')$
      \begin{equation} \label{eq:mhmat}
        \Psi(i,i') = \begin{cases} \min\{\frac{1}{N_i}, \frac{1}{N_{i'}}\}, ~\text{for}~\forall i'\in\mathcal N_i; \\ 1-\sum_{k \in\mathcal N_i} \Psi(i,k), ~\text{for}~i=i'. \end{cases}
      \end{equation}
      The matrix $\Psi$ (with $\Psi(i,i')$ being the $(i,i')$-th component) is the so-called \textit{transition matrix} of the MHRW. Let $q^t_{i,i'}$ denote the probability that the walk initialized by agent $i$ researches agent $i'$ after $t$ steps. When the graph $\mathcal G$ is connected and non-bipartite, $\Psi$ is a symmetric doubly stochastic matrix such that the random walk (initialized by $i$) has a unique uniform stationary distribution with $q^t_{i,i'} = \frac{1}{N}$ for $\forall i'\in\mathcal N$ when $t \rightarrow \infty$. According to \cite{LevinP-book17}, the random walk achieves a  $\alpha$-nearly uniform distribution (such that $\frac{1}{N} - \alpha \leq q^t_{i,i'} \leq \frac{1}{N} + \alpha$ for $\forall i' \in \mathcal N$) in at most $\frac{1}{\Gamma(\Psi)}\log\left(\frac{2N}{\alpha}\right)$ steps, where $\Gamma(\Psi)$ denotes the spectral gap of the transition matrix $\Psi$. In the following, to facilitate our algorithm analysis, we let $\alpha = \frac{1}{N^3}$ without sacrificing the generality and rationality of our analysis. Since $\Gamma(\Psi)$ does not depend on $N$, it is said that a MHRW achieves a nearly uniform distribution in $\mathcal O(\log N)$ steps. In fact, $\alpha$ can be made as small as required at the expense of constant.

      %Let $q^t_{i,i'}$ denote the probability that the walk initialized by agent $i$ researches agent $i'$ after $t$ steps. It is said that, the random walk achieve a $\epsilon$-nearly uniform probability distribution if $ \frac{1}{N} - \epsilon \leq q^t_{i,i'} \leq \frac{1}{N} + \epsilon$ for $\forall i' \in \mathcal N$. As shown in \cite{YuanLYYWLC-ICDCS19}, $\mathcal O(\log N)$ steps are sufficient for the MHRW to achieve the $\frac{1}{N^3}$-nearly uniform distribution in any $\alpha$-expander graph.

    \subsubsection{Local Differential Privacy} \label{sssec:ldp}
    %
      %During the learning process, the agents need to communicate with each other for exchanging private information (i.e., their latest adoptions in our case), which entails a high demand on privacy preserving. In this paper, we consider a strong setting of local privacy by assuming that each agent trusts none of the others. For example, for each agent, if one of its neighbors is compromised by an adversary, the adversary may be able to eavesdrop on the private information shared by the agent. Therefore, 
      %
      In this paper, we leverage the notion of \emph{Local Differential Privacy} (LDP) for the purpose of privacy preserving.
      \begin{definition}[Local Differential Privacy \cite{Kasiviswanathan-JOC11,DuchiJW-FOCS13}] \label{def:ldp}
        Let $\varepsilon$ be a positive real number and $\mathcal F$ be a randomized algorithm which takes a user's private data set $\mathcal D$ as input. Let $\mathsf{im}\mathcal D$ be the image of the algorithm $\mathcal F$. The algorithm $\mathcal F$ is said to be able to deliver $\varepsilon$-differential privacy, if for any pair of the user's possible private data $x, x'\in \mathcal D$ and any subset $\mathcal K$ of $\mathsf{im}\mathcal D$,
        \begin{equation} \label{eq:ldp}
          \frac{\mathbb P [\mathcal{F}(x)\in \mathcal K]}{\mathbb P [\mathcal{F}(x')\in \mathcal K]}\le \exp(\varepsilon).
        \end{equation}
      \end{definition}

      $\varepsilon$ is the so-called \textit{privacy budget}, which specifies the privacy loss which we can afford. Specifically, we have to allow more loss of privacy with a higher privacy budget. Especially, when $\epsilon=\infty$, the randomized algorithm cannot offer any privacy preservation.

    \subsubsection{Basic Facts} \label{sssec:facts}
      We hereby introduce a few theorems and definitions which will be useful in our algorithm analysis. %We give the well-known Chernoff bound, Hoeffding inequality and Bernoulli inequality in \textbf{Theorem}~\ref{thm:chbd}, \textbf{Theorem}~\ref{thm:hoeff} and \textbf{Theorem}~\ref{thm:berineq}, respectively. We then introduce two inequalities which will be used later in \textbf{Theorem}~\ref{thm:t33}. We finally define a notion $\sim$
      \begin{theorem}[Chernoff Bounds~\cite{DubhashiP-book09}] \label{thm:chbd}
        Let $Z_1, \cdots, Z_N$ be independent Bernoulli random variables with $\mathbb E[Z_i]=\pi_i$. Assume $\pi = \frac{1}{N} \sum^N_{i=1} \pi_i$. When $0 < \tau \leq 1$, we have
        \begin{equation*}
          \mathbb P\left[ \left| \frac{1}{N} \sum^N_{i=1} Z_i - \pi \right| > \pi\tau \right] \leq 2\exp\left( \frac{-N\pi\tau^2}{3} \right).
        \end{equation*}
        Specifically, when $Z_1, \cdots, Z_N$ are i.i.d. random variables such that $\pi_i = \pi$ for $\forall i=1,\cdots,N$, we have
        \begin{equation*}
          \mathbb P\left[ \frac{1}{N} \sum^N_{i=1} Z_i \geq (1+\tau) \pi \right] \leq \exp\left( - \frac{N\pi\tau^2}{3} \right), ~\forall \tau>0
        \end{equation*}
        %
        %for $\forall \tau > 0$.
        %
      % 
      \end{theorem}
      \begin{theorem}[Hoeffding Inequality~\cite{DubhashiP-book09}] \label{thm:hoeff}
        Let $Z_1, \cdots, Z_N$ be independent Bernoulli random variables with $\mathbb E[Z_i]=\pi_i$. Assume $\pi = \frac{1}{N} \sum^N_{i=1} \pi_i$. For $\forall \tau >0$, we have
        \begin{equation*}
          \mathbb P\left[ \left| \frac{1}{N} \sum^N_{i=1} Z_i - \pi \right| \geq \tau \right] \leq 2\exp\left( -2N\tau^2 \right).
        \end{equation*}
      \end{theorem}

      \begin{theorem}[Bernoulli Inequality] \label{thm:berineq}
        Supposing $z>-1\neq 0$ is a real number and $n>1$ is an integer, we have $(1+z)^n > 1+nz$.
      \end{theorem}

      \begin{theorem} \label{thm:t33}
        Assuming $Z \geq 2$ is a positive integer, we have
        \begin{numcases}{}
          \left( 1 - \frac{2}{Z^2} \right) > \left( 1 - \frac{2}{Z^3-Z^2} \right)^Z \label{eq:t33-1} \\
          \left(1+\frac{2}{Z^3-Z^2-1}\right)^Z>\left(1+\frac{2}{Z^2-1}\right)
          \label{eq:t33-2}
        \end{numcases}
      \end{theorem}
      \begin{proof}
        Let $z=-\frac{2}{Z^3-Z^2}$, then $-1<z<0$ since $Z\ge2$. We then have
        \begin{flalign}
            &1-\left(1-\frac{2}{Z^3-Z^2}\right)^Z
            =1-(1+z)^Z\nonumber\\
            =&[ 1-(1+z)]\left[1+(1+z)+(1+z)^2+\cdots+(1+z)^{Z-1}\right]\nonumber\\
            %=&(-z)\left[1+(1+z)+(1+z)^2+\cdots+(1+z)^{Z-1}\right]\nonumber\\
            >&(-z)\left[1+(1+z)+(1+2z)+\cdots+(1+(Z-1)z)\right]\nonumber\\
            =&(-z)\{Z+z[1+2+\cdots+(Z-1)]\}\nonumber\\
            %=&\frac{2}{Z^3-Z^2}\left(Z-\frac{2}{Z^3-Z^2}\cdot\frac{Z(Z-1)}{2}\right)\nonumber\\
            =&\frac{2}{Z^2(Z-1)}\cdot\left(Z-\frac{1}{Z}\right)>\frac{2}{Z^2(Z-1)}\cdot(Z-1)=\frac{2}{Z^2} \nonumber
        \end{flalign}
        where the first inequality is due to the Bernoulli Inequality (see \textbf{Theorem}~\ref{thm:berineq}). (\ref{eq:t33-1}) can be obtained by rearranging the above inequality. Note that $\frac{2}{Z^3-Z^2-1}>0$ when $Z\ge2$. The inequality (\ref{eq:t33-2}) follows the Bernoulli Inequality directly:
        %
        % \begin{flalign*}
        %   &\left(1+\frac{2}{Z^3-Z^2-1}\right)^Z>1+\frac{2Z}{Z^3-Z^2-1}\nonumber\\
        %   =&1+\frac{2}{Z^2-Z-1/Z}
        %   >1+\frac{2}{Z^2-Z}>1+\frac{2}{Z^2-1}.
        % \end{flalign*}
        %
        \begin{align*}
          \left(1+\frac{2}{Z^3-Z^2-1}\right)^Z>1+\frac{2Z}{Z^3-Z^2-1}>1+\frac{2}{Z^2-1}
        \end{align*}
      \end{proof}

      \begin{definition} \label{def:funcg}
        We denote by $g: \mathbb N^+ \rightarrow \mathbb R$ a function such that
        \begin{enumerate}[(i)]
          \item For any real number $\ell>0$, there exists a positive integer $N_1\geq 2$ such that $g(N) > \ell\ln N$ for $\forall N \geq N_1$.
          \item For any real number $\ell>0$, there exists a positive integer $N_2\geq 2$ such that $g(N) < \ell N$ for ~$\forall N \geq N_2$
        \end{enumerate}
      \end{definition}

      \begin{definition} \label{def:sim}
        Given real numbers $Z_1$, $Z_2$ and $c \geq 1$, the notation $Z_1 \overset{c}{\sim} Z_2$ denotes $\frac{1}{c} \leq \frac{Z_1}{Z_2} \leq c$. 
      \end{definition}

\section{Algorithm} \label{sec:alg}
  In this section, we present our privacy-preserving social learning algorithm in a general graph. Our algorithm proceeds iteratively in rounds and each agent performs the following four stages in each round. In \textbf{Stage 1}, each agent perturbs the option it adopted in the last round, for the purpose of privacy preserving. Then, the agents disseminate their perturbed adoptions over the network through MHRWs in \textbf{Stage 2}. Thereafter, in \textbf{Stage 3}, each agent may receive a number of perturbed adoptions from its peers, according to which, the agent selects one option as a candidate. The candidate option is then considered to be adopted or not in \textbf{Stage 4}, based on the most recent observation on its stochastic quality. In the following, we present the details of the four stages, respectively.

  \subsection{Stage 1: Perturbing} \label{ssec:pert}
    Let $\mathbf{X}^{r-1}_i = \left[ X^{r-1}_{i,1}, \cdots, X^{r-1}_{i,M} \right]$ be the adoption vector of agent $i$ in round $r-1$. If agent $i$ does not adopt any option in round $r-1$ such that $X^{r-1}_{i,j} = 0$ for $\forall j\in\mathcal M$, it does nothing to the adoption vector $\mathbf{X}^{r-1}_i$ (and thus the variables $\{X^{r-1}_{i,j}\}_{j\in\mathcal M}$) in round $r$; otherwise, $i$ perturbs its adoption vector $\mathbf{X}^{r-1}_i$ according to a perturbing mechanism $\widetilde{\mathbf X}^{r}_i = \mathcal F(\mathbf{X}^{r-1}_i)$, where $\mathbf{\widetilde X}^{r}_i = \left[ \widetilde{X}_{i,1}^{r}, \cdots,  \widetilde{X}_{i,M}^{r}\right]$ denote the perturbed adoption vector of agent $i$ in round $r$. Specifically, we design the perturbing mechanism $\mathcal F$ as follows
    \begin{equation}
    \label{eq:pert}
      \widetilde{X}_{i,j}^{r}=
      \begin{cases}
        X_{i,j}^{r-1} & \mbox{with probability $\frac{\exp(\varepsilon/2)}{\exp(\varepsilon/2)+1}$};\\
        1-X_{i,j}^{r-1} & \mbox{with probability $\frac{1}{\exp(\varepsilon/2)+1}$}.
      \end{cases}
    \end{equation}
    In another word, for any option $j$, each agent $i$ ``flips'' $X_{i,j}^{r-1}$ with probability $\frac{1}{\exp({\varepsilon/2})+1}$. We denote by $\widetilde{\mathcal X}^r = \{\mathbf{\widetilde X}^{r}_i\}_{i\in\mathcal N}$ the output of this stage, i.e., the set of the perturbed adoption vectors in round $r$.

  \subsection{Stage 2: Disseminating} \label{ssec:diss}
    In this stage, we employ MHRWs to disseminate the perturbed adoption vectors $\widetilde{\mathcal X}^r$ over graph $\mathcal G$. Each vector is associated with a \emph{length} variable indicating the maximum times it is forwarded in a random walk. The vector and its length variable (as well as some prerequisite information specified by specific communication protocols) are encapsulated in a data token. A token is said to be \emph{feasible} if it has a non-zero length variable.

    For any agent $i$ adopting some option in round $r-1$, it sets off $hg(N)$ MHRWs in parallel (where $h = \frac{16\sigma}{1-\beta}$ with $\sigma \geq 11$). Each random walk has a length of $\mathcal O(\log N)$. The agent $i$ uses a \textit{First-in-First-out} (FIFO) queue to buffer the tokens (with non-zero length indicators) received to forward next. In each slot of round $r$, $i$ pops the first (up to) $hg(N)$ feasible tokens out of the queue, and forwards each of the tokens to either one of its neighbors or itself according to the probability distribution $\Psi(i,i')$ (see Eq.~(\ref{eq:mhmat})). The lengths of the tokens are decreased by one before the forwarding. A token (and thus a perturbed adoption vector) is said to be ``sampled'' by a agent if it reaches the agent with the associated length variable being zero. Let $\mathcal V^r_i$ denote the set of the perturbed adoption vectors sampled by agent $i$ and $V^r_i = |\mathcal V^r_i|$ be the size of $\mathcal V^r_i$. Note that, an agent may receive multiple perturbed adoption vectors from the same one.

    \subsection{Stage 3: Sampling} \label{ssec:samp}
      In this stage, each agent $i$ selects an option to consider in the following adopting stage. Specifically, the agent $i$, with probability $\mu$, selects an option $j$ uniformly at random \footnote{We hereby use the probability $\mu$ to force the agents to ``explore'' the options, to prevent our algorithm from getting stuck in a local optimum. In practice, the parameter $\mu>0$ is usually small.}; with probability $1-\mu$, $i$ selects one of the options according to their \emph{normalized} popularity estimates $\{\widehat{Q}^{r}_{i,j}\}_{j\in\mathcal M}$. In particular, agent $i$ first estimates $j$'s (unnormalized) popularity $\widetilde Q^{r}_{i,j}$ according to the sampled perturbed adoption vectors $\mathcal V^{r}_i$ in round $r$ as
      %
      % \begin{equation}
      % \label{eq:estpop}
      %   \widetilde{Q}_{i,j}^{r-1} = \frac{\exp(\varepsilon/2)+1}{\exp(\varepsilon/2)-1}\frac{\sum_{i'\in\mathcal N^r_i} \widetilde{X}^{r-1}_{i',j}}{N^r_i}-\frac{1}{\exp(\varepsilon/2)-1}
      %   %&=& \frac{\exp(\varepsilon/2)+1}{\exp(\varepsilon/2)-1} \cdot \frac{\sum_{k=1}^{n_i}V_{kj}^{i,r}}{n_i}-\frac{1}{e^{\varepsilon/2}-1}
      % \end{equation}
      %
      \begin{equation}
      \label{eq:estpop}
        \widetilde{Q}_{i,j}^{r} =  \max\left\{ \frac{\exp\left({\varepsilon}/{2}\right)+1}{\exp\left({\varepsilon}/{2}\right)-1}\Lambda^r_{i,j}-\frac{1}{\exp\left({\varepsilon}/{2}\right)-1}, 0\right\},
        %&=& \frac{\exp(\varepsilon/2)+1}{\exp(\varepsilon/2)-1} \cdot \frac{\sum_{k=1}^{n_i}V_{kj}^{i,r}}{n_i}-\frac{1}{e^{\varepsilon/2}-1}
      \end{equation}
      where $\Lambda^{r}_{i,j} = \sum_{\mathbf{\widetilde X}\in\mathcal V^r_i} [\mathbf{\widetilde X}]_j / V^{r}_i$ and $[\widetilde{\mathbf X}]_j$ denotes the $j$-th element of the vector $\widetilde{\mathbf X}\in\mathcal V^{r}_i$. $\widetilde{Q}_{i,j}^{r}$ is then normalized by 
      %$\widehat{Q}^{r}_{i,j} = {\widetilde{Q}^{r}_{i,j}} \Big/ {\sum^m_{j'=1}  \widetilde{Q}^{r}_{i,j'} }$
      %
      \begin{equation} \label{eq:norm}
        \widehat{Q}^{r}_{i,j} = {\widetilde{Q}^{r}_{i,j}} \Big/ {\sum^m_{j'=1}  \widetilde{Q}^{r}_{i,j'} }
      \end{equation}
      
      such that $\widehat{Q}_{i,j}^{r} \geq 0~\text{and}~\sum^m_{j=1} \widehat{Q}_{i,j}^{r} = 1$.

    \subsection{Stage 4: Adopting} \label{ssec:adopt}
      Let $j^*$ be the option sampled by agent $i$ in the above stage. The agent $i$ then decides whether or not to adopt the option $j^*$ according to the following rule:
      \begin{equation}
      \label{eq:adoprule}
        X^r_{i,j^*} = \begin{cases}1~~\text{with probability}~\beta~\text{if}~\Phi^r_{j^*}=1;\\
        1~~\text{with probability}~1-\beta~\text{if}~\Phi^r_{j^*}=0 ;\\
        0~~\text{otherwise}.\end{cases}
      \end{equation}
      where $\beta > 1/2$ and is close to $1/2$.
      In particular, if observing the most recent quality signal $\Phi^r_{j^*}=1$, with probability $\beta$, the agent $i$ adopts the option $j^*$ such that $X^r_{i,j^*} = 1$ and $X^r_{i,j} = 0$ for $\forall j\neq j^*$, while with probability $1-\beta$, $i$ does not adopt any option such that $X^r_{i,j} = 0$ for $\forall j$. If $\Phi^r_{j^*}=0$ is observed, $i$ adopts $j^*$ with probability $1-\beta$ or adopts none of the options with probability $\beta$.

\section{Analysis} \label{sec:analysis}
  In this section, we present the details of our analysis on the algorithm. We first analyze our algorithm from the perspective of communication complexity in Sec.~\ref{ssec:complexity}. We then demonstrate how our algorithm has the regret function bounded in Sec.~\ref{ssec:converge} and finally discuss the privacy preservation of our algorithm in Sec.~\ref{ssec:prip}.

  \subsection{Communication Complexity} \label{ssec:complexity}
    It is demonstrated above that our algorithm entails very light-weight computations; therefore, we hereby concentrate on revealing the communication complexity in each round (i.e., the number of slots in each round for disseminating perturbed adoption vectors), while postponing the analysis on the number of rounds our algorithm takes to achieve convergence in Sec.~\ref{ssec:converge}. 

    As shown in Sec.~\ref{ssec:diss}, each agent launches $h g(N)$ MHRWs in each round in the disseminating stage, the question is, given that each MHRW entails $\mathcal O(\log N)$ steps to achieve a nearly uniform distribution, how many slots are necessitated in our case to ensure all MHRWs approach the nearly uniform distributions. Although this question has been (partially) answered in a quite different context in our previous work~\cite{YuanLYYWLC-ICDCS19}, we hereby provide a sketch of our specialized answer in \textbf{Theorem}~\ref{thm:mhrw}.
    \begin{theorem} \label{thm:mhrw}
      Consider a connected non-bipartite graph consisting of a sufficiently large number of agents such that $g(N) > \ln N$. When each agent sets off $h g(N)$ MHRWs, with probability at least $1-\frac{1}{N^{h/3}}$, all the MHRWs achieve a nearly uniform distribution $\left[ \frac{1}{N}-\frac{1}{N^3},\frac{1}{N}+\frac{1}{N^3}\right]$ within $\mathcal O(\log^2N)$ slots.
    \end{theorem}
    \begin{proof}
      According to \textbf{Stage 2: Disseminating} in Sec.~\ref{sec:alg}, for each agent $i$, the expected number of the tokens it receives from its neighbors in each slot is
      \begin{eqnarray*}
        && \sum_{i'\in\mathcal N_i} \min\left\{ \frac{1}{N_i}, \frac{1}{N_{i'}} \right\} \times h\cdot g(N) \\
        &=& \sum_{i'\in\mathcal N_i: \frac{1}{N_i} \geq \frac{1}{N_{i'}}} \frac{h\cdot g(N)}{N_{i'}} + \sum_{i'\in\mathcal N_i: \frac{1}{N_i} \leq \frac{1}{N_{i'}}} \frac{h\cdot g(N)}{N_{i}} \\
        &\leq& \sum_{i'\in\mathcal N_i: \frac{1}{N_i} \geq \frac{1}{N_{i'}}} \frac{h\cdot g(N)}{N_{i}} + \sum_{i'\in\mathcal N_i: \frac{1}{N_i} \leq \frac{1}{N_{i'}}} \frac{h\cdot g(N)}{N_{i}} \\
        &=& N_i \times \frac{h\cdot g(N)}{ N_i} = h\cdot g(N).
      \end{eqnarray*}
      By applying the Chernoff bound (see \textbf{Theorem}~\ref{thm:chbd}), agent $i$ receives at most $2h\cdot g(N)$ tokens in each slot with probability at least $1-\frac{1}{N^{h/3}}$, when $N$ is sufficiently large such that $g(N) > \ln N$. Furthermore, considering we employ a FIFO forwarding policy, the tokens agent $i$ receives in some slot $t$ can be delayed for at most $t$ additional slots. As a token should be forwarded for $\mathcal O(\log N)$ times to achieve a nearly uniform distribution, we conclude that with probability at least $1-N^{-h/3}$, all MHRWs in round $r$ approach a $\frac{1}{N^3}$-nearly uniform distribution in $\Delta = \mathcal O(\log^2 N)$ slots, according to what we have shown in Sec.~\ref{sssec:rw}.
    \end{proof}
    
    \begin{remark} \label{rem:complexity}
      It follows \textbf{Theorem}~\ref{thm:mhrw} that, in our disseminating stage, each agent needs to send $\mathcal O(g(N)\log^2N)$ messages, each of which has $\mathcal O(\log N)$ bits \footnote{More precisely, the length of the message should be $\mathcal O(M+\log N)$ (in bits). Nevertheless, throughout our analysis, we focus on investigating how large $N$ should be given fixed $M$. In this sense, we consider $M$ is a constant.}. Therefore, we conclude that our algorithm has a per-round communication complexity of $\mathcal O(g(N)\log^3 N)$ for each agent. According to our definition of $g(N)$ in \textbf{Definition}~\ref{def:funcg}, the complexity can be re-written (with a slight relaxation) as $\mathcal O(N\log^3 N)$.
    \end{remark}

  \subsection{Convergence} \label{ssec:converge}
    The challenges for analyzing the convergence of our algorithm are two-fold: on one hand, although our MHRW-based disseminating stage entails very efficient communications, each agent gets random (and thus incomplete) suggestions from its peers; on the other hand, the agents introduce random perturbations to their their private adoptions for the purpose of privacy preservation, such that the experience each agent learns from their peers is noisy. As \cite{CelisKV-PODC17} has shown the convergence of the learning dynamics by assuming each agent in round $r$ is aware of the actual popularities of all the options, i.e., $\{Q^{r-1}_{j}\}_{j\in\mathcal M}$, our focus is at demonstrating that, for each agent $i$ in round $r$, its estimate on $Q^{r-1}_{j}$, namely $\widehat Q^{r}_{i,j}$, sufficiently approximates $Q^{r-1}_{j}$ for $\forall j$. In the following, we first present our main results on the convergence of the regret function and then give the detailed proof.

  \subsubsection{Main Result} \label{ssec:mainres}
    \begin{theorem} \label{thm:main}
      Assume there are $M$ unknown options such that $\eta_1\ge\eta_2\ge\cdots\ge\eta_M$. Let $\frac{1}{2} < \beta < \frac{e}{e+1}$ and define $\delta =\ln\left( \frac{\beta}{1-\beta}\right)$ (hence $\delta>0$ and $\delta\to0$). Suppose $\varepsilon>0$, $6\mu \leq \delta^2$, $h=\frac{16\sigma}{1-\beta}$ (with $\sigma \geq 11$) and
      \begin{equation} \label{eq:constc}
        c=\frac{4M(2M+1)}{\mu(1-\beta)}\cdot\frac{\exp(\varepsilon/2)+1}{\exp(\varepsilon/2)-1}.
      \end{equation}
      When $N$ is sufficiently large such that
      \begin{eqnarray} \label{eq:population}
        \begin{cases}
        \frac{c^2M^{\frac{2\ln 5}{\delta^2}}}{\delta^2}\ln N<g(N)<N; \\
        N>\max\left\{ \frac{8h}{\sqrt{5}}+\sqrt{2}, \left(\frac{(10M+3)\ln M}{\delta^3}\right)^{\frac{1}{10}}\right\}
        \end{cases}
      \label{eq:conN}
      \end{eqnarray}
      for any $R$ such that $\frac{1}{\delta^2}{\ln \frac{4M}{\mu(1-\beta)}} \leq R \leq \frac{N^{10}\delta}{6M}$, we have
      \begin{equation} \label{eq:regdyn}
        \mathsf{Regret}_N(R):=\eta_1-\frac{1}{R}\sum_{r=1}^{R}\sum_{j=1}^{M}\mathbb{E}\left[ Q_j^{r-1}\Phi_j^r\right]\le 6\delta.  
      \end{equation}
    \end{theorem}

    \begin{remark} \label{re:main}
      In the above theorem, we give the answers to the fundamental questions proposed in Sec.~\ref{sec:intro}. It is revealed in \textbf{Theorem}~\ref{thm:main} that, when there are a sufficiently large number of agents engaged in our social learning process, there exists a constant upper bound on the regret of our algorithm with finite time horizon, even the experience shared by each agent is perturbed for the purpose of privacy preserving. In fact, \textbf{Theorem}~\ref{thm:main} implies a trade-off among the number of the agents (thus the communication overhead as shown in Sec.~\ref{ssec:complexity}), privacy preserving and learning utility. In particular, we could have higher learning utility (and thus smaller regret) while guaranteeing the local privacy for each agent, if more agents participate in the social learning process, resulting in higher communication overhead. Furthermore, given a certain number of agents participating in the social learning process, if there are more unknown options to learn or less privacy loss is allowed, we have to be content with a sacrifice in learning utility (and thus increased regret). We will perform extensive simulations to verify the trade-off later in Sec.~\ref{sec:exp}.
    \end{remark}

  \subsubsection{Detailed Proof} \label{ssec:deanaly}
  %
    %Let $D_j^r = \sum^N_{i=1} X^r_{i,j}$ denote the number of agents adopting option $j$ in round $r$, then the popularity of option $j$ in round $r$, namely $Q_j^r$, can be rewritten as $Q_j^r=\frac{D_j^r}{\sum_{k=1}^{m}D_k^r}$.

    As shown in Sec.~\ref{ssec:samp}, in round $r$, each agent $i$ selects an option as a candidate in the sampling stage according to $\mu$ and $\{\widehat{Q}^{r}_{i,j}\}_{j\in\mathcal M}$. Assuming $Y_{i,j}^{r} \in \{0,1\}$ is a random variable indicating if agent $i$ selects option $j$ in the sampling stage in round $r$, the probability of $Y_{i,j}^{r}=1$ conditioned on $\widehat Q^{r}_{i,j}$ can be defined as
    \begin{equation}
      \mathbb{P}\left[Y_{i,j}^{r}=1 \Big| \widehat{Q}_{i,j}^{r}\right] =(1-\mu)\widehat{Q}_{i,j}^{r}+\frac{\mu}{m}
    \end{equation}
    and we thus have
    \begin{equation}
      \mathbb{E}\left[Y_{i,j}^{r} \Big| \widehat{Q}_{i,j}^{r}\right] = (1-\mu)\widehat{Q}_{i,j}^{r}+\frac{\mu}{m}\ge\frac{\mu}{m}.
    \end{equation}
    Let $\mathcal{S}_j^{r} = \{ i\in\mathcal N \mid Y^r_{i,j} = 1 \}$ denote the set of the agents that select option $j$ in the sampling stage of round $r$ and $S_j^{r} = \sum^N_{i=1} Y_{i,j}^{r}$ be the size of $\mathcal S_j^{r}$. We calculate the conditional expectation of $S_j^{r}$ as follows
    \begin{align}
      \mathbb{E}\left[ S_j^{r} \big| \widehat{Q}_{i,j}^{r} \right] &= \mathbb{E}\left[ \sum_{i=1}^{N}Y_{i,j}^{r} \Bigg| \widehat{Q}_{i,j}^{r}\right]  = \sum\limits_{i=1}^{N} \left((1-\mu)\widehat{Q}_{i,j}^{r}+\frac{\mu}{m}\right) \nonumber\\
      &= \left((1-\mu)\widehat{Q}_{j}^{r}+\frac{\mu}{m}\right)N
      \geq \frac{\mu N}{m};
    \end{align}
    where $\widehat{Q}_j^{r} = \frac{1}{N} \sum^N_{i=1} \widehat{Q}^{r}_{i,j}$. Given $\widehat{Q}_j^{r}$, according to the Chernoff bound, the following lemma holds as a straightforward extension of \textbf{Proposition} 4.6 in \cite{CelisKV-PODC17}.
    \begin{lemma} \label{le:sq}
     In each round $r$, for each option $j$ , with probability at least $1-\frac{2M}{N^{10}}$ (conditioned on $\widehat Q^{r}_{j}$)
     \begin{equation} \label{eq:sq}
       S_j^{r}\overset{1+2\delta'}{\sim}\left( (1-\mu)\widehat{Q}_j^{r}+\frac{\mu}{M}\right)N;
     \end{equation}
     where $\delta'=\sqrt{\frac{30M\ln N}{\mu N}}\le\frac{1}{2}$. Moreover, for $\forall j\in\mathcal M$, $S_j^{r}\ge\frac{\mu N}{2M}$, 
     %
     %\begin{equation}
     %  S_j^{r}\ge\frac{\mu N}{2M}
     %\end{equation}
     %
     with probability at least $1-\frac{2m}{N^{10}}$.
    \end{lemma}

    As shown in Sec.~\ref{sec:alg}, the probability of $X_{i,j}^{r}=1$ (conditioned on $\mathcal{S}_j^{r}$ and $\Phi_j^{r}$) can be defined as
    \begin{equation}
      \mathbb{P}\left[ X_{i,j}^{r}=1 \Big| \mathcal{S}_j^{r},\Phi_j^{r}\right]=
      \begin{cases}
        \beta^{\Phi_j^{r}}(1-\beta)^{1-\Phi_j^{r}} & \mbox{if $i\in\mathcal{S}_j^{r}$};\\
        0 & \mbox{otherwise}.\nonumber
      \end{cases}
    \end{equation}
    Supposing $D_j^{r} = \sum^N_{i=1} X_{i,j}^{r}$ denotes the number of agents adopting option $j$ in round $r$, it follows that
    \begin{equation}
      \mathbb{E}\left[ D_j^{r}\Big| \mathcal{S}_j^{r},\Phi_j^{r}\right]=S_j^{r}\beta^{\Phi_j^{r}}(1-\beta)^{1-\Phi_j^{r}}.
    \end{equation}
    Especially, due to our assumption $\beta>\frac{1}{2}$, $\beta^{\Phi_j^{r}}(1-\beta)^{1-\Phi_j^{r}}\ge(1-\beta)$.
    It follows \cite{CelisKV-PODC17} again that
    \begin{lemma} \label{le:ds}
      In any round $r$, for any option $j$, with probability at least $1-\frac{4M}{N^{10}}$ (conditioned on $\mathcal{S}_j^{r},\Phi_j^{r}$ ),
      \begin{equation}
        D_j^{r}\overset{1+2\delta''}{\sim}S_j^{r}\beta^{\Phi_j^{r}}(1-\beta)^{1-\Phi_j^{r}};
      \end{equation}
      where $\delta''=\sqrt{\frac{60M\ln N}{\mu N}}\le\frac{1}{2}$.
    \end{lemma}

    \textbf{Lemma}~\ref{le:sq} and  \textbf{Lemma}~\ref{le:ds} characterize the relationship between $S^r_j$ and $\widehat Q^r_j$ and the one between $S^r_j$ and $D^r_j$ with the notation of ``$\sim$'', respectively. By combining them, we derive the relationship between $D_j^{r}$ and $\widehat{Q}_j^{r}$ in \textbf{Lemma}~\ref{le:dq2}.
    \begin{lemma} \label{le:dq2}
      In any round $r$, for any option $j$, with probability at least $1-\frac{6M}{N^{10}}$ (conditioned on $\Phi_j^{r}$ and $\widehat{Q}^{r}_j$ ),
      \begin{equation}
        D_j^{r}\overset{1+6\delta''}{\sim}\left((1-\mu)\widehat{Q}_j^{r}+\frac{\mu}{m}\right)N\beta^{\Phi_j^{r}}(1-\beta)^{1-\Phi_j^{r}}.
      \end{equation}
      %
      %where $\delta''=\sqrt{\frac{60m\ln N}{\mu N}}\le\frac{1}{2}$
    \end{lemma}
    \begin{proof}
      The proof follows directly from \textbf{Lemma}~\ref{le:sq} and \textbf{Lemma}~\ref{le:ds}, by noticing that $(1+2\delta')(1+2\delta'')\le1+2\delta'+2\delta''+4\delta'\delta''\le1+6\delta''$ as $\delta'\le\delta''\le\frac{1}{2}$.
    \end{proof}

    Let $D^r = \sum^M_{j=1} D_j^{r}$ denote the number of the agents with non-null adoptions. Based on the relationship between $D_j^{r}$ and $\widehat{Q}_j^{r}$ shown in \textbf{Lemma}~\ref{le:dq2}, we derive the lower bounds for $D^{r}$ and $Q_j^{r}={D_j^{r}}\big/{\sum_{k=1}^{M}D_k^{r}}$ in \textbf{Lemma}~\ref{le:qd} as follows.
    \begin{lemma} \label{le:qd}
      For any round $r$ and option $j$, with probability at least $1-\frac{6M}{N^{10}}$,
      \begin{equation}
        D^{r} \geq \frac{N(1-\beta)}{4} ~~\text{and}~~Q_j^{r}\ge\frac{\mu(1-\beta)}{4M}.
      \end{equation}
    \end{lemma}
    \begin{proof}
      According to \textbf{Lemma}~\ref{le:dq2}, with probability at least $1-\frac{6M}{N^{10}}$, we have
      \begin{eqnarray*}
        D_j^{r} &\ge& \frac{1}{1+6\delta''}\left((1-\mu)\widehat{Q}_j^{r}+\frac{\mu}{m}\right)N\beta^{\Phi_j^{r}}(1-\beta)^{1-\Phi_j^{r}} \\
        &\geq& \frac{(1-\beta)N}{4} \left((1-\mu)\widehat{Q}_j^{r}+\frac{\mu}{m}\right) \geq \frac{(1-\beta)N\mu}{4m}
      \end{eqnarray*}
      by considering $\delta'' \leq \frac{1}{2}$ and $\beta > \frac{1}{2}$. Moreover, as $\sum^M_{j=1} D^{r}_j \leq N$, $Q_j^{r}$ can be re-written as
      \begin{equation*}
        Q_j^{r}=\frac{D_j^{r}}{\sum_{k=1}^{M}D_k^{r}}\ge\frac{\frac{\mu(1-\beta)}{4M}N}{N}=\frac{\mu(1-\beta)}{4M}.
      \end{equation*}
      Similarly, with probability at least $1-\frac{6M}{N^{10}}$, we deduce that
	  \begin{eqnarray*}
	    D^{r} &\ge& \sum_{j=1}^{M} \frac{1}{1+6\delta''}\left((1-\mu)\widehat{Q}_j^{r}+\frac{\mu}{M}\right)N\beta^{\Phi_j^{r}}(1-\beta)^{1-\Phi_j^{r}}\\
	    &\ge& \sum_{j=1}^{M}\frac{1}{4}\left((1-\mu)\widehat{Q}_j^{r}+\frac{\mu}{M}\right)N(1-\beta)\\
	    %&=& \frac{1}{4}N(1-\beta)\sum_{j=1}^{m}\left((1-\mu)\widehat{Q}_j^{r-1}+\frac{\mu}{m}\right)\\
	    &=& \frac{1}{4}N(1-\beta)\left((1-\mu)\sum_{j=1}^{M}\widehat{Q}_j^{r}+ \mu \right) = \frac{(1-\beta)N}{4},
	  \end{eqnarray*}
	  where the equality in the last step is due to the fact that $\sum^m_{j=1} \widehat{Q}^{r}_j = \frac{1}{N}\sum^N_{i=1}\sum^m_{j=1}{\widehat{Q}^{r}_{i,j}}=1$.
	  %
	  % \[ \sum^m_{j=1} \widehat{Q}^{r}_j = \frac{1}{N}\sum^N_{i=1}\sum^m_{j=1}{\widehat{Q}^{r}_{i,j}}=1 \]
    \end{proof}

    Given $N$ fixed, \textbf{Lemma}~\ref{le:qd} indicates there are at least $\frac{N(1-\beta)}{4}$ agents which have non-null adoptions in each round $r-1$ and at least $\frac{\mu(1-\beta)}{4M}$ of them adopting option $j$ with high probability. Let ${\mathcal V}^{r}$ be the set of all perturbed adoption vectors disseminated in round $r$ \footnote{As shown in Sec.~\ref{ssec:diss}, ${\mathcal V}^{r}$ is formed by replicating each agent's perturbed adoption vector for $hg(N)$ times.} and ${\mathcal V}^r_i \subseteq {\mathcal V}^{r}$ be a subset which agent $i$ samples (or receives) in round $r$. We also suppose $V^r = D^rhg(N)$ and $V^r_i$ denotes the size of $\mathcal V^{r}$ and the one of ${\mathcal V}^r_i$, respectively. In \textbf{Lemma}~\ref{le:ubound}, we show the upper and lower bounds of $V^r_i$, with the help of \textbf{Theorem}~\ref{thm:mhrw} and \textbf{Lemma}~\ref{le:qd}.
    \begin{lemma} \label{le:ubound}
      For each agent $i$, when $N$ is sufficiently large such that $N \geq 4\ln N$, the number of the sampled perturbed adoption vectors in round $r$, i.e., $V^r_i$, satisfies
      \begin{equation} \label{eq:vecperround}
        \frac{3(1-\beta)}{32}hg(N)\le V^r_i\le \frac{15}{8}hg(N)
      \end{equation}
      with a probability at least $1 - \frac{6M+3}{N^{10}}$.
    \end{lemma}
    \begin{proof}
      Assume $Z^r_{i,k}\in\{0,1\}$ is a Bernoulli random variable indicating if the $k$-th token in round $r$ arrives at agent $i$ through the random walk. $V^r_i$ then can be represented as $V^r_i = \sum^{D^rhg(N)}_{k=1} Z^r_{i,k}$. Let $\gamma^r_{i,k} = \mathbb E[Z^r_{i,k}]$ and $\gamma^r_i = \frac{1}{D^rhg(N)} \sum^{D^rhg(N)}_{k=1} \gamma^r_{i,k}$. To facilitate our presentation, we defined the following two events $\mathsf{E_1}(i,r) = \left\{\frac{1}{N}-\frac{1}{N^3} \leq \gamma^r_{i,k} \leq \frac{1}{N}+\frac{1}{N^3}, \forall k \right\}$ and $\mathsf{E_2}(r) = \left\{D^r \geq \frac{(1-\beta)N}{4} \right\}$, and thus $\mathbb P[\mathsf{E_1}(i,r)] \geq 1- \frac{1}{N^{h/3}}$ (see \textbf{Theorem}~\ref{thm:mhrw}) and $\mathbb P[\mathsf{E_2}(r)] \geq 1-\frac{6M}{N^{10}}$ (see \textbf{Lemma}~\ref{le:qd}).
      According to the Chernoff-Hoeffding bound (see \textbf{Theorem}~\ref{thm:chbd}),
      \begin{align*}
        & \mathbb P\left[ \left| \frac{V^r_i}{D^rhg(N)} - \gamma^r_i \right| \geq \gamma^r_i \sqrt{\frac{\ln N}{g(N)}} ~\Bigg|~ \mathsf{E_1}(i,r), \mathsf{E_2}(r) \right] \\
        %
        %& \leq 2\exp\left(-\frac{\ln N}{3g(N)} \sum^{D^rhg(N)}_{k=1} \gamma^r_{i,k}\right) \\
        %
        %& \leq 2\exp\left(-\frac{1}{3}{\left(\frac{1}{N}-\frac{1}{N^3}\right) D^r h \ln N}\right) \\
        %
        &\leq 2\exp\left(-\frac{1}{12}(1-\frac{1}{N^2})(1-\beta)h\ln N\right) \leq \frac{2}{N^{\sigma}}
      \end{align*}
      % %
      %
      where we have the second inequality when the two events both hold and the last inequality due to the fact that $1-\frac{1}{N^2} \geq \frac{3}{4}$ when $N \geq 4\ln N$. Taking the union bound across $i\in\mathcal N$, we get 
      \begin{align*}
        &\mathbb P\left[ \left| \frac{V^r_i}{D^r h g(N)} - \gamma^r_i \right| \leq \gamma^r_i \sqrt{\frac{\ln N}{g(N)}}, \forall i\in\mathcal N \bigg|~ \mathsf{E_1}(i,r), \mathsf{E_2}(r) \right] \\
        & \geq 1-\frac{2}{N^{\sigma-1}}
      \end{align*}
      Since $\mathbb P[\mathsf{E_1}(i,r)] \geq 1- \frac{1}{N^{h/3}}$ and $\mathbb P[\mathsf{E_2}(r)] \geq 1-\frac{6M}{N^{10}}$,
      \begin{eqnarray*}
        \mathbb P\left[ \left| \frac{V^r_i}{D^r h g(N)} - \gamma^r_i \right| \leq \gamma^r_i \sqrt{\frac{\ln N}{g(N)}}, \forall i \right] > 1-\frac{6M+3}{N^{10}}
      \end{eqnarray*}
      When $N$ is sufficiently large such that $g(N) \geq 4 \ln N$, $\left| \frac{V^r_i}{D^r h g(N)} - \gamma^r_i \right| \leq \gamma^r_i \sqrt{\frac{\ln N}{g(N)}}$ can be re-written as the inequality (\ref{eq:vecperround}), by considering the facts $\frac{(1-\beta)N}{4} \leq D^r \leq N$ and $\gamma^r_{i,k} \in \left[\frac{1}{N}-\frac{1}{N^3}, \frac{1}{N}+\frac{1}{N^3} \right]$.
    \end{proof}

    Let ${\Lambda}^{r}_j = \sum_{\mathbf{\widetilde X} \in \mathcal{V}^{r}} [\mathbf{\widetilde X}]_j / V^r$ denote the the fraction of perturbed adoption vectors which indicate option $j$ is adopted in round $r$. We also assume that $\Lambda^{r}_{i,j} = \sum_{\mathbf{\widetilde X}\in\mathcal V^r_i} [\mathbf{\widetilde X}]_j / V^{r}_i$ denotes the fraction of the perturbed adoption vectors received by agent $i$ in round $r$, which indicate option $j$ is adopted. In another word, $\Lambda^{r}_{j}$ is the perturbed popularity of option $j$ in round $r$, while $\Lambda^{r}_{i,j}$ is an estimate of agent $i$ on $\Lambda^{r}_{j}$. Given the nearly uniform disseminating distribution (see \textbf{Theorem}~\ref{thm:mhrw}) and the range of $V^r_i$ (see \textbf{Lemma}~\ref{le:ubound}), we show the upper and lower bounds of ${\mathbb E\left[\Lambda^r_{i,j} \big| V^r_i, \widetilde{\mathcal X}^{r}\right]} \big{/} {\Lambda^{r}_j}$ in \textbf{Lemma}~\ref{le:hypergeo} and \textbf{Lemma}~\ref{le:xi01}.
    \begin{lemma}\label{le:hypergeo}
      Given the $\frac{1}{N^3}$-nearly uniform distribution in round $r$ (with probability at least $1-\frac{1}{N^{h/3}}$), the following inequality holds for $\forall i\in\mathcal N$ and $\forall j\in \mathcal M$
	    \begin{equation} \label{eq:hypergeo}
	      \xi_0 \leq \frac{\mathbb E\left[\Lambda^r_{i,j} \big| V^r_i, \widetilde{\mathcal X}^{r}\right]}{\Lambda^{r}_j} \leq \xi_1
	    \end{equation}
      where
      \begin{equation} \label{eq:xi01}
        \begin{cases}
          \xi_0 = \left(1-\frac{2}{N^2+1}\right)^{V^r_i} \left(1-\frac{2}{N^3-N^2+1}\right)^{V^r-V^r_i} \\
          \xi_1 = \left(1+\frac{2}{N^2-1}\right)^{V^r_i} \left(1+\frac{2}{N^3-N^2-1}\right)^{V^r-V^r_i}
        \end{cases}
      \end{equation}
    \end{lemma}
    \begin{proof}
      Considering the nearly uniform distribution (with probability at least $1-\frac{1}{N^{h/3}}$), conditioned on ${\mathcal V}^{r}$, the probability that agent $i$ samples $V^r_i = v^r_i$ tokens in round $r$, $\mathbb P(V^r_i=V \mid \widetilde{\mathcal X}^{r})$, can be bounded by
      %
      % \begin{equation*}
      %   \xi^{'}_0 \leq \mathbb P(V^r_i = v^r_i \mid {\mathcal V}^{r}) \leq \xi^{'}_1
      % \end{equation*}
      %
      \begin{equation*}
        \dbinom{V^r}{v^r_i} \xi'_0 \leq \mathbb P(V^r_i = v^r_i \mid {\mathcal V}^{r}) \leq \dbinom{V^r}{v^r_i} \xi'_1
      \end{equation*}
      where
      \begin{eqnarray*}
        \begin{cases}
          \xi'_0 = \left(\frac{1}{N}-\frac{1}{N^3}\right)^{v^r_i}\left(1-\frac{1}{N}-\frac{1}{N^3}\right)^{V^r-v^r_i}\\
          \xi'_1 = \left(\frac{1}{N}+\frac{1}{N^3}\right)^{v^r_i}\left(1-\frac{1}{N}+\frac{1}{N^3}\right)^{V^r-v^r_i}
        \end{cases}
      \end{eqnarray*}
	    Furthermore, letting $B^r_{i,j} = \Lambda^r_{i,j} V^r_i = \sum_{\mathbf{\widetilde X}\in \mathcal V^r_i}[\mathbf{\widetilde X}]_j$ denote the number of the perturbed adoption vectors sampled by agent $i$ in round $r$ with the $j$-th component being $1$, we then have $\mathbb P(B^r_{i,j}=b^r_{i,j}, V^r_i=v^r_i \mid {\mathcal V}^{r})$ lie in the range of
      \begin{eqnarray*}
        \left[ \dbinom{\Lambda^r_j{V}^{r}}{b^r_{i,j}}\dbinom{(1-\Lambda^r_j){V}^{r}}{v^r_i-b^r_{i,j}}\xi'_0, ~~\dbinom{\Lambda^r_j{V}^{r}}{b^r_{i,j}}\dbinom{(1-\Lambda^r_j){V}^{r}}{v^r_i-b^r_{i,j}}\xi'_1 \right]
      \end{eqnarray*}
	    by considering the nearly uniform distribution resulting from the random walk-based dissemination again. Moreover, due to
	    \[\mathbb P(B^r_{i,j}=b^r_{i,j} \mid V^r_i=v^r_i,\mathcal V^{r}) = \frac{\mathbb P(B^r_{i,j}=b^r_{i,j}, V^r_i=v^r_i \mid \mathcal V^{r})}{\mathbb P(V^r_i=v^r_i \mid \mathcal V^r)}\]
	    we have $\zeta_0 \leq \mathbb P(B^r_{i,j}=b^r_{i,j} \mid V^r_i=v^r_i,\mathcal V^{r}) \leq \zeta_1$ where
	    \begin{eqnarray*}
	    \begin{cases}
	    \zeta_0 = \frac{\xi'_0}{\xi'_1} \cdot {\dbinom{\Lambda^r_j{V}^{r}}{b^r_{i,j}}\dbinom{(1-\Lambda^r_j){V}^{r}}{v^r_i-b^r_{i,j}}}{\dbinom{V^r}{v^r_i}}^{-1} \\
	    \zeta_1 = \frac{\xi'_1}{\xi'_0} \cdot {\dbinom{\Lambda^r_j{V}^{r}}{b^r_{i,j}}\dbinom{(1-\Lambda^r_j){V}^{r}}{v^r_i-b^r_{i,j}}}{\dbinom{V^r}{v^r_i}}^{-1} \\
	    \end{cases}
	    \end{eqnarray*}
	    Therefore, we have $\mathbb E\left[B^r_{i,j} \big| V^r_i, \widetilde{\mathcal X}^{r}\right]$ bounded by $\frac{\xi'_0}{\xi'_1} \cdot \Lambda^{r}_j V^r_i \leq \mathbb E\left[B^r_{i,j} \mid V^r_i, {\mathcal V}^{r}\right] \leq \frac{\xi'_1}{\xi'_0} \cdot \Lambda^{r}_j V^r_i$ and finally complete the proof by considering the fact that $\mathbb E\left[\Lambda^r_{i,j} \mid V^r_i, {\mathcal V}^{r}\right] = {\mathbb E\left[B^r_{i,j} \mid V^r_i, {\mathcal V}^{r}\right]}/{V^r_i}$.
    \end{proof}

    \begin{lemma} \label{le:xi01}
      Let $N$ be sufficiently large such that $g(N)<N$, $(N - 1)\sqrt{\frac{\ln N}{N}}\ge\frac{8h}{\sqrt{5}}$, and $\left(N-\frac{3}{2}\right)\cdot\ln \left(1+\sqrt{\frac{5\ln N}{N}}\right) \ge 8h$. We have the following two inequalities hold in any round $r$ with probability at least $1-\frac{6M+3}{N^{10}}$,
      \begin{eqnarray}
        \xi_0 \ge 1-\sqrt{\frac{5\ln N}{N}} ~\text{and}~ \xi_1 \le 1 + \sqrt{\frac{5\ln N}{N}}
      \end{eqnarray}
    \end{lemma}
    \begin{proof}
      Recalling $V^r_i\le \frac{15hg(N)}{8}$ with probability at least $1-\frac{6M+3}{N^{10}}$ (as shown in \textbf{Lemma \ref{le:ubound}}), we have
      \begin{eqnarray*}
        \xi_0 &\ge& \left(1-\frac{2}{N^2}\right)^{V^r_i}\left(1-\frac{2}{N^3-N^2}\right)^{V^r} \nonumber\\
        &\ge& \left(1-\frac{2}{N^2}\right)^{2hg(N)}\left(1-\frac{2}{N^3-N^2}\right)^{N\cdot2hg(N)} \nonumber\\
        &\ge& \left(1-\frac{2}{N^3-N^2}\right)^{4hNg(N)} \ge \left(1-\frac{2}{N^3-N^2}\right)^{4hN^2} \nonumber\\
        &\ge& 1 - \frac{8h}{N-1}
      \end{eqnarray*}
      where we have the third inequality by applying the inequality (\ref{eq:t33-1}) (shown in \textbf{Theorem}~\ref{thm:t33}), the forth one due to the fact that $g(N) < N$, and the last one by applying the Bernoulli inequality (see \textbf{Theorem}~\ref{thm:berineq}). Furthermore, since $(N - 1)\sqrt{\frac{\ln N}{N}} \geq \frac{8h}{\sqrt{5}}$, we have $\xi_0 \ge 1-\sqrt{\frac{5\ln N}{N}}$ holds with probability at least $1-\frac{6M+2}{N^{10}}$.

      Likewise, when $V^r_i\le \frac{15hg(N)}{8}$ holds, we deduce that
      \begin{eqnarray}
        \xi_1 &\le& \left(1+\frac{2}{N^2-1}\right)^{V_i^r}\left(1+\frac{2}{N^3-N^2-1}\right)^{V^r} \nonumber\\
        %
        % &\le& \left(1+\frac{2}{N^2-1}\right)^{\frac{15}{8}hg(N)}\left(1+\frac{2}{N^3-N^2-1}\right)^{Nhg(N)} \nonumber\\
        %
        &\le& \left(1+\frac{2}{N^2-1}\right)^{2hN}\left(1+\frac{2}{N^3-N^2-1}\right)^{2hN^2} \nonumber\\
        &\le& \left(1+\frac{2}{N^3-N^2-1}\right)^{4hN^2} \le \exp\left( \frac{8h}{N-{3}/{2}} \right) \nonumber
      \end{eqnarray}
      where we have the last inequality due to the fact that $\left(1+\frac{1}{N}\right)^N\le e$ holds for any positive integer $N$. Furthermore, $\xi_1 - 1 \leq \exp\left( \frac{8h}{N-{3}/{2}} \right) - 1 \leq \sqrt{\frac{5\ln N}{N}}$, since $\left(N-\frac{3}{2}\right)\cdot\ln \left(1+\sqrt{\frac{5\ln N}{N}}\right) \geq 8h $.
    \end{proof}

    It is shown in the above two lemmas that , through the disseminating stage of our algorithm, each agent $i$ can accurately estimate the perturbed popularity of any option $j$ in each round $r$. In the following \textbf{Lemma}~\ref{le:hatq}, we demonstrate that agent $i$ also can accurately estimate the actual popularity of each option $j$ in round $r$, by showing the absolute difference between $\widehat{Q}^r_j$ and $Q^{r-1}_j$ is bounded.
    \begin{lemma}\label{le:hatq}
      Conditioned on the adoptions $\mathcal X^{r-1} = \{\mathbf{X}^{r-1}_i\}_{i\in\mathcal N}$ in round $r-1$, for each option $j\in\mathcal M$ in round $r$, with probability at least $1-\frac{10M+3}{N^{10}}$, 
      \begin{equation}\label{eq:hatq}
        \big|\widehat{Q}_j^{r}-Q_j^{r-1}\big| \le 6(2M+1)\sqrt{\frac{\ln N}{g(N)}}\frac{\exp(\varepsilon/2)+1}{\exp(\varepsilon/2)-1}
      \end{equation}
      where $\widehat{Q}^{r}_j = \frac{1}{N} \sum^N_{i=1} \widehat{Q}^{r}_{i,j}$.
      %
      % \begin{equation} \label{eq:zeta}
      %   \zeta= \frac{\exp(\varepsilon/2)+1}{\exp(\varepsilon/2)-1} \left(\sqrt{\frac{\ln N}{3g(N)}}+2\sqrt{\frac{5\ln N}{N}}\right)
      % \end{equation}
      % and
    %   \begin{equation*}
    %     \widehat{Q}^{r}_j = \frac{1}{N} \sum^N_{i=1} \widehat{Q}^{r}_{i,j}.
    %   \end{equation*}
      %
    \end{lemma}
    \begin{proof}
      According to perturbation process shown in Sec.~\ref{ssec:pert}, the Bernoulli probability for any ${\mathbf X}\in\mathcal X^{r-1}$ to have $[\widetilde{\mathbf X}]_j=1$ (where $\widetilde{\mathbf X}$ denotes $\mathbf X$'s perturbed counterpart) is 
      \begin{equation*}
        \mathbb P[[\widetilde{\mathbf X}]_j=1 \mid \mathcal X^{r-1}] = \begin{cases}
          \frac{\exp(\varepsilon/2)}{\exp(\varepsilon/2)-1}, ~~[{\mathbf X}]_j=1 \\
          \frac{1}{\exp(\varepsilon/2)-1}, ~~[{\mathbf X}]_j=0
        \end{cases}
      \end{equation*}
      Furthermore, since each agent $i$ does not carry out the perturbation in round $r$ if it has no option adopted in round $r-1$, we have
      \begin{align*}
        \mathbb E[\Lambda^r_j \mid \mathcal X^{r-1}] &= \frac{1}{D^r}\sum_{\mathbf X \in \mathcal X^{r-1}} \mathbb E[[\widetilde{\mathbf X}]_j \mid \mathcal X^{r-1}] \\
        %
        %&= \frac{Q_j^{r-1}\cdot\exp(\varepsilon/2)}{\exp({\varepsilon/2})+1}+\frac{1-Q_j^{r-1}}{\exp(\varepsilon/2)+1} \\
        %
        &= \frac{Q_j^{r-1} \left( \exp(\varepsilon/2)-1 \right) +1 }{\exp(\varepsilon/2)+1}
      \end{align*}
  	  Applying the Hoeffding's inequality (see \textbf{Theorem}~\ref{thm:hoeff}) and the union bound, we have
      \begin{align} \label{eq:LambdaQ}
        &\mathbb P \Bigg[\left|\Lambda^{r}_{j}-\frac{Q_j^{r-1} \left( \exp(\varepsilon/2)-1 \right) +1 }{\exp(\varepsilon/2)+1}\right| \leq \sqrt{\frac{5\ln N}{N}}, \forall j \Bigg|~ \mathcal X^{r-1} \Bigg] \nonumber\\
        &\geq 1- \frac{2M}{N^{10}}
      \end{align}

      Similarly, supposing there exists $c^r_{i,j} \in [\xi_0, \xi_1]$ for $\forall i,j$ according to \textbf{Lemma}~\ref{le:hypergeo} and \textbf{Lemma}~\ref{le:xi01} such that $\mathbb{E}\left[\Lambda^{r}_{i,j} \bigg| V^r_i, \widetilde{\mathcal X}^{r}\right]=c^r_{i,j} \Lambda^{r}_{j}$, we can deduce that
      \begin{equation}\label{eq:lam}
        \mathbb P\left[ \left|\Lambda^{r}_{i,j}-c^r_{i,j}\cdot\Lambda^{r}_{j}\right| \leq \sqrt{\frac{11\ln N}{2V^r_i}}, \forall i,j ~\Bigg|~ V^r_i, \widetilde{\mathcal X}^{r}\right]\ge1-\frac{2M}{N^{10}}
      \end{equation}
      As we have shown in \textbf{Lemma}~\ref{le:ubound} that with probability at least $1-\frac{6M+3}{N^{10}}$, $V^r_i \geq \frac{3(1-\beta)}{32}hg(N)$ such that for any agent $i$ and option $j$, with probability at least $1-\frac{8M+3}{N^{10}}$ (conditioned on $\widetilde{\mathcal X}^{r}$), we have
      \begin{equation}
        \left|\Lambda^{r}_{i,j}-c^r_{i,j}\cdot\Lambda^{r}_{j}\right|\le\sqrt{\frac{11\ln N}{2V^r_i}} \le \sqrt{\frac{\ln N}{3g(N)}}
      \end{equation}
      and hence,
      \begin{eqnarray} \label{eq:Lambda2ij}
        \left|\Lambda^{r}_{i,j}-\Lambda^{r}_{j}\right|
        %
        %&=& \left|\Lambda^{r}_{i,j}-c^r_{i}\cdot\Lambda^{r}_{j}+c^r_{i}\cdot\Lambda^{r}_{j}-\Lambda^{r-1}_{j}\right| \nonumber\\
        %
        &\le& \left|\Lambda^{r}_{i,j}-c^r_{i}\cdot\Lambda^{r}_{j}\right|+\left|c^r_{i}\cdot\Lambda^{r}_{j}-\Lambda^{r}_{j}\right| \nonumber\\
        &\le& \left|\Lambda^{r}_{i,j}-c^r_{i}\cdot\Lambda^{r}_{j}\right|+\left|c^r_{i}-1\right| \nonumber\\
        &\le& \sqrt{\frac{\ln N}{3g(N)}}+\max\left\{\ \left| \xi_0-1\right|,\left| \xi_1-1\right|\right\} \nonumber\\
        &\le& \sqrt{\frac{\ln N}{3g(N)}}+\sqrt{\frac{5\ln N}{N}},
      \end{eqnarray}
      %
      % i.e.
      % %
      % \begin{equation}\label{eq:bdlam}
      %   \mathbb P\left[ \left|\Lambda^r_{i,j}-\Lambda^r_j \right|\le\sqrt{\frac{\ln N}{3g(N)}}+\sqrt{\frac{5\ln N}{N}}\Bigg| \widetilde{\mathcal X}^{r}\right]\ge 1-\frac{8M+2}{N^{10}}
      % \end{equation}
      %
      Since
      \begin{eqnarray*}
        && \left|\Lambda^{r}_{i,j} -  \frac{Q_j^{r-1} \left( \exp(\varepsilon/2)-1 \right) +1 }{\exp(\varepsilon/2)+1} \right| \\
        %
        %&=& \left|\Lambda^{r}_{i,j} - \Lambda^{r}_{j} + \Lambda^{r}_{j} -  \frac{Q_j^{r-1} \left( \exp(\varepsilon/2)-1 \right) +1 }{\exp(\varepsilon/2)+1} \right| \\
        %
        &\leq& \left|\Lambda^{r}_{i,j}-\Lambda^{r}_{j}\right| + \left|\Lambda^{r}_{j}-  \frac{Q_j^{r-1} \left( \exp(\varepsilon/2)-1 \right) +1 }{\exp(\varepsilon/2)+1}\right| \\
        %
        %&\le& \sqrt{\frac{\ln N}{3g(N)}}+2\sqrt{\frac{5\ln N}{N}},
      \end{eqnarray*}
      we have
      \begin{align}\label{eq:basebd}
        &\mathbb P\bigg[ \bigg|\Lambda^{r}_{i,j} - \frac{Q_j^{r-1} \left( \exp(\varepsilon/2)-1 \right) +1 }{\exp(\varepsilon/2)+1} ~\bigg|~ \nonumber\\
        &~~~\leq \sqrt{\frac{\ln N}{3g(N)}}+2\sqrt{\frac{5\ln N}{N}} ~\bigg|~ \mathcal X^{r-1}\bigg]\ge1-\frac{10M+3}{N^{10}}
      \end{align}
      by combining (\ref{eq:LambdaQ}) and (\ref{eq:Lambda2ij}). It then follows that
      \begin{eqnarray}
        \mathbb P\left[\left| \widetilde Q^{r}_{i,j} - Q^{r-1}_j \right|\le \zeta \bigg| \mathcal X^{r-1}\right]\nonumber \geq 1-\frac{10M+3}{N^{10}}
      \end{eqnarray}
      where
      \begin{eqnarray*} \label{eq:zeta}
        \zeta &=& \frac{\exp(\varepsilon/2)+1}{\exp(\varepsilon/2)-1} \left(\sqrt{\frac{\ln N}{3g(N)}}+2\sqrt{\frac{5\ln N}{N}}\right) \\
        &\leq& 6\sqrt{\frac{\ln N}{g(N)}}\frac{\exp(\varepsilon/2)+1}{\exp(\varepsilon/2)-1}
      \end{eqnarray*}
      with $N$ being sufficiently large such that $g(N) \leq N$, since
      \begin{eqnarray*}
        &&\left| \widetilde Q^{r}_{i,j} - Q^{r-1}_j \right|\nonumber\\
        &\le&\left| \frac{\exp(\varepsilon/2)+1}{\exp(\varepsilon/2)-1}\Lambda_{i,j}^{r}-\frac{1}{\exp(\varepsilon/2)-1}-Q_j^{r-1}\right|\nonumber\\
        &\le&\frac{\exp(\varepsilon/2)+1}{\exp(\varepsilon/2)-1}\left| \Lambda_{i,j}^{r} - \frac{Q_j^{r-1}(\exp(\varepsilon/2)-1)+1}{\exp(\varepsilon/2)+1} \right|\nonumber\\
        %
        %&\le&\left(1+\frac{2}{e^{\varepsilon/2}-1}\right)\left(\sqrt{\frac{\ln N}{3g(N)}}+2\sqrt{\frac{5\ln N}{N}}\right),
      \end{eqnarray*}
      When $\left| \widetilde Q^{r}_{i,j} - Q^{r-1}_j \right| \le \zeta$ such that $1-M\zeta \leq \sum^M_{j=1} \widetilde{Q}^r_{i,j} \leq 1+M\zeta$, we have
      \begin{align}\label{eq:zeta_bd}
        & \left| \widehat Q^{r}_{i,j} - \widetilde Q^{r}_{i,j} \right| = \left| \frac{\widetilde Q^{r}_{i,j}}{\sum_{j'=1}^{M}\widetilde Q^{r}_{i,j'}}- \widetilde Q^{r}_{i,j}\right| \nonumber\\
        &\le \max\left(\frac{\widetilde Q^{r}_{i,j}}{1-M\zeta}-\widetilde Q^{r}_{i,j}, \widetilde Q^{r}_{i,j}-\frac{\widetilde Q^{r}_{i,j}}{1+M\zeta}\right)\nonumber\\
        %
        %&= \widetilde Q^{r}_{i,j}\cdot\max\left(\frac{1}{1-M\zeta}-1, 1-\frac{1}{1+M\zeta}\right)\nonumber\\
        %
        &\le \max\left(\frac{M\zeta}{1-M\zeta}, \frac{M\zeta}{1+M\zeta}\right) \le \frac{M\zeta}{1-M\zeta} \le 2M\zeta
      \end{align}
      where the last inequality holds when $M\zeta \leq 1/2$ (which we will prove in \textbf{Lemma}~\ref{le:proximity}). Hence, for $\forall i, j$, with probability at least $1-\frac{10M+3}{N^{10}}$ (conditioned on $\mathcal X^{r-1}$),
      \begin{align}
        & \left| \widehat Q^{r}_{i,j} - Q^{r-1}_j \right| = \left| \widehat Q^{r}_{i,j} - \widetilde Q^{r}_{i,j}+\widetilde Q^{r}_{i,j}-Q^{r-1}_j\right| \nonumber\\
        &\le \left| \widehat Q^{r}_{i,j} - \widetilde Q^{r}_{i,j} \right|+ \left| \widetilde Q^{r}_{i,j}-Q^{r-1}_j\right| \le (2M+1)\zeta.
      \end{align}
      and thus $\left| \widehat Q^{r}_{j} - Q^{r-1}_j \right| = \left| \frac{1}{N}\sum_{i=1}^{N}\widehat Q^{r}_{i,j} - Q^{r-1}_j \right| \leq \frac{1}{N}\sum_{i=1}^{N}\left| \widehat Q^{r}_{i,j} - Q^{r-1}_j\right| \leq  (2M+1)\zeta$. We finally complete the proof by substituting (\ref{eq:zeta}) into the above one.
    \end{proof}

    As shown in the following \textbf{Lemma}~\ref{le:proximity}, the relationship between $\widehat{Q}^r_j$ and $Q^{r-1}_j$ under the notation ``$\sim$'' then can be derived from \textbf{Lemma}~\ref{le:qd} and \textbf{Lemma}~\ref{le:hatq}.
    \begin{lemma}\label{le:proximity}
      Let $N$ be sufficiently large such that $\sqrt{\frac{g(N)}{\ln N}} \geq c\frac{M^{\frac{\ln 5}{\delta^2}}}{\delta}$ where $c=\frac{4M(2M+1)}{\mu(1-\beta)}\frac{\exp(\varepsilon/2)+1}{\exp(\varepsilon/2)-1}$ and $\delta = \ln\left(\frac{\beta}{1-\beta}\right)$. In any round $r$, with probability at least $1-\frac{10M+3}{N^{10}}$, we have
      \begin{equation}
        \widehat{Q}_j^{r}\overset{1+2\delta_Q}{\sim}Q_j^{r-1}
      \end{equation}
      for any option $j$, where $\delta_Q= c\sqrt{\frac{\ln N}{g(N)}}$.
      %
      % \begin{equation}
      %   \delta_Q= c\sqrt{\frac{\ln N}{g(N)}}
      % \end{equation}
    \end{lemma}
    \begin{proof}
      From \textbf{Lemma}~\ref{le:qd} and \textbf{Lemma}~\ref{le:hatq}, it is shown that
      \begin{equation*}
        {\left| \widehat Q_j^r-Q_j^{r-1} \right|}\big/ {Q_j^{r-1}} \le \delta_Q
      \end{equation*}
      Therefore, $(1-\delta_Q)Q_j^{r-1}\le\widehat{Q}_j^r\le(1+\delta_Q)Q_j^{r-1}$. Furthermore, since $\frac{M^{\frac{\ln 5}{\delta^2}}}{\delta} \geq 2$ for $0<\delta \leq 1$, $\delta_Q \leq \frac{1}{2}$ (and thus $M\zeta<1/2$ holds in (\ref{eq:zeta_bd})). Thus $\frac{1}{1+2\delta_Q}\le(1-\delta_Q)$ and it follows that $\frac{1}{1+2\delta_Q}Q_j^{r-1}\le\widehat{Q}_j^r\le(1+2\delta_Q)Q_j^{r-1}$.
    \end{proof}
    \begin{remark}
      As shown by \textbf{Definition}~\ref{def:funcg}, we have $\sqrt{\frac{\ln N}{g(N)}}$ approach zero when $N$ becomes infinity. Therefore, when $N\rightarrow\infty$ given fixed $c$, in any round $r$, the average estimate on option $j$'s popularity, i.e., $\widehat Q^r_j$, approaches the actual popularity of the option $j$, namely $Q^{r-1}_j$, as closely as possible.
    \end{remark}

    We now are ready to prove our main result $\textbf{Theorem}~\ref{thm:main}$ following the thread shown in \cite{CelisKV-PODC17}. In particular, we investigate the dynamics of our learning algorithm by coupling it with the \textit{Multiplicative Weights Update} (MWU) method. The MWU method is a very powerful tool in a wide spectrum of learning and optimization problems. By defining a weight $W^r_j$ for $\forall j, r$ as follows
    \begin{equation}
      W_j^{r+1} = \left( (1-\mu)W_j^t+\frac{\mu}{m}\sum_{k=1}^{m}W_k^t\right)\beta^{\Phi_j^{t+1}}(1-\beta)^{1-\Phi_j^{t+1}}
    \end{equation}
    with $W^0_j = 1$ for $\forall j$, we get a probability distribution $P^r_j = \frac{W^r_j}{\sum^M_{j'=1} W^r_j}$. As demonstrated in \textbf{Lemma}~\ref{le:pqrelation}, for $\forall j, r$, $Q^r_j$ approaches $P^r_j$ as closely as possible, especially when there are a infinite number of agents.
    \begin{lemma} \label{le:pqrelation}
      Let $\delta_r = 5^r \delta_Q$. For any option $j$ in round $r$, $P_j^r\overset{1+\delta_r}{\sim}Q_j^{r}$ holds with probability at least $1-\frac{(10M+3)r}{N^{10}}$ for all choices of ${\Phi_j^r}$'s.
    \end{lemma}
    \begin{proof}
      The proof proceeds by the inducting on $r$. It is apparent that $P^0_j = Q^0_j$ for $\forall j\in\mathcal M$. We assume that the statement holds for $\forall r \geq 1$, such that $P^r_j \overset{1+\delta_r}{\sim} Q^r_j$ with probability at least $1-\frac{(10M+3)r}{N^{10}}$ for each option $j$ in round $r$. Since
      \begin{equation*}
        P_j^{r+1}=\frac{\left( (1-\mu)P_j^r+\frac{\mu}{M}\right)\beta^{\Phi_j^{r+1}}(1-\beta)^{1-\Phi_j^{r+1}}} {\sum_{j'=1}^{M}\left( (1-\mu)P_{j'}^r+\frac{\mu}{M}\right)\beta^{\Phi_{j'}^{r+1}}(1-\beta)^{1-\Phi_{j'}^{r+1}}}
      \end{equation*}
      with probability at least $1-\frac{(10M+3)r}{N^{10}}$,
      \begin{equation*}
        P_j^{r+1}\overset{(1+\delta_r)^2}{\sim}\hspace{-2ex}\frac{\left( (1-\mu)Q_j^r+\frac{\mu}{M}\right)\beta^{\Phi_j^{r+1}}(1-\beta)^{1-\Phi_j^{r+1}}} {\sum_{j'=1}^{M}\left( (1-\mu)Q_{j'}^r+\frac{\mu}{M}\right)\beta^{\Phi_{j'}^{r+1}}(1-\beta)^{1-\Phi_{j'}^{r+1}}}
      \end{equation*}
      Furthermore, according to \textbf{Lemma}~\ref{le:dq2} and \textbf{Lemma}~\ref{le:proximity}, we deduce that
      \begin{equation}
        P_j^{r+1}\overset{(1+\delta_r)^2(1+6\delta'')^2(1+2\delta_Q)^2}{\sim}\frac{D_j^{r+1}}{\sum_{j'=1}^{m}D_{j'}^{r+1}}=Q_j^{r+1}
      \end{equation}
      with probability at least $1-\frac{(10M+3)(r+1)}{N^{10}}$. Assuming that $\delta_r=5^r\delta_Q\le1$, $\delta_Q\le\frac{1}{27}$ and $\delta''\le\frac{1}{6}\delta_Q$, we have $(1+\delta_r)^2(1+6\delta'')^2(1+2\delta_Q)^2 \leq 1+\delta_{r+1}$ hold for $\forall r\ge2$. For $r=1$ the bound can be checked by a direct calculation.
      %
      % \begin{eqnarray*}
      %   && (1+\delta_r)^2(1+6\delta'')^2(1+2\delta_Q)^2 \\
      %   &\le& (1+\delta_r)^2(1+\delta_Q)^2(1+2\delta_Q)^2 \\
      %   &\le& (1+3\delta_r)(1+3\delta_Q)(1+5\delta_Q) \\
      %   &=& (1+3\delta_r)(1+8\delta_Q+15\delta_Q^2) \\
      %   &\le& (1+3\delta_r)(1+9\delta_Q)
      %   \le 1+3\delta_r+9\tau+27\delta_Q\delta_r\\
      %   &\le& 1+4\delta_r+9\delta_Q
      %   < 1+4\delta_r+5^2\delta_Q\\
      %   &\le& 1+4\delta_r+5^r\delta_Q
      %   = 1+5\delta_r = 1+\delta_{r+1},
      % \end{eqnarray*}
    %
    \end{proof}

    As shown in \cite{CelisKV-PODC17}, let $\frac{1}{2} < \beta \leq \frac{e}{e+1}$ (and thus $0 < \delta \leq 1$) and $6\mu \leq \delta^2$. With infinite population (i.e., $N=\infty$) and uniform initialization $P^0_j = \frac{1}{M}$ for $\forall j$, for $T \geq \frac{\ln M}{\delta^2}$,
    \vspace{-2ex}
    \begin{eqnarray} \label{eq:regmwu}
      \eta_1 - \frac{1}{R}\sum^R_{r=1} \sum^M_{j=1} \mathbb E[P^{t-1}_j \Phi^r_j] \leq 3\delta
    \end{eqnarray}
    Especially, when $P^0_j \geq \tau$ (for $\forall j\in\mathcal M$) is non-uniform, the inequality (\ref{eq:regmwu}) still holds for $R \geq \frac{\ln(1/\tau)}{\delta^2}$.

    According to the (stochastic) coupling between $P^r_j$ and $Q^r_j$ shown in \textbf{Lemma}~\ref{le:proximity}, we can deduce that
    \begin{eqnarray*}
      && \frac{1}{R}\sum^R_{r=1} \sum^M_{j=1} \mathbb E[P^{t-1}_j \Phi^r_j]\\
      &\leq& \frac{1}{R}\sum^R_{r=1} \left(1-\frac{(10M+3)r}{N^{10}}\right)(1+5^r\delta_Q) \sum^M_{j=1} \mathbb E[Q^{t-1}_j \Phi^r_j] \\
      && + \frac{1}{R}\sum^R_{r=1}\frac{(10M+3)r}{N^{10}}\sum^M_{j=1} \mathbb E[P^{t-1}_j \Phi^r_j] \\
      &\leq& (1+5^R\delta_Q)\frac{1}{R}\sum^R_{r=1}\sum^M_{j=1} \mathbb E[Q^{t-1}_j \Phi^r_j] + \frac{(10M+3)R}{N^{10}} 
    \end{eqnarray*}
    Therefore,
    \begin{eqnarray}
      \eta_1 - \frac{1}{R}\sum^M_{j=1} \mathbb E[Q^{t-1}_j \Phi^r_j] \leq 3\delta + 5^R\delta_Q + \frac{(10M+3)R}{N^{10}}
    \end{eqnarray}
    When $R=\frac{\ln M}{\delta^2}$, $5^R\delta_Q \leq M^{\frac{\ln 5}{\delta^2}}c\sqrt{\frac{\ln N}{g(N)}}$. Therefore, when $N$ is sufficiently large such that
    \begin{eqnarray}
      \sqrt{\frac{g(N)}{\ln N}} \geq \frac{c M^{\frac{\ln5}{\delta^2}}}{\delta}  ~\text{and}~ N^{10} \geq \frac{(10M+3)\ln M}{\delta^3}
    \end{eqnarray}
    we can deduce that
    \begin{equation} \label{eq:ncond}
      \eta_1 - \frac{1}{R}\sum^M_{j=1} \mathbb E[Q^{t-1}_j \Phi^r_j] \leq 3\delta + 5^R\delta_Q + \frac{(10M+3)R}{N^{10}} \leq 5\delta
    \end{equation}
    %
    %which partially completes the proof of \textbf{Theorem}~\ref{thm:main} with $R = \frac{\ln M}{\delta^2}$.

    The above result can be extended to handle non-uniform initiation by letting $N$ be sufficiently large such that $\sqrt{\frac{g(N)}{\ln N}} \geq \frac{c (1/\tau)^{\frac{\ln5}{\delta^2}}}{\delta}$ and $N^{10} \geq \frac{(10/\tau+3)\ln(1/\tau)}{\delta^3}$ (instead of (\ref{eq:ncond})). When $T = \frac{\ln(1/\tau)}{\delta^2}$, we have $\eta_1 - \frac{1}{R}\sum^M_{j=1} \mathbb E[Q^{t-1}_j \Phi^r_j] \leq 5\delta$. Therefore, when $R > \frac{\ln M}{\delta^2}$, we can break the time into epochs, each of which consists of $\frac{\ln(1/\tau)}{\delta^2}$ rounds. In each epoch, we then have the regret function upper-bounded regardless of whether or not the initial distribution is uniform. Specifically, it is demonstrated in \textbf{Lemma}~\ref{le:qd} that, $Q^r_j \geq \frac{\mu(1-\beta)}{4M}$ for $\forall j$ with probability at least $1-\frac{6M}{N^{10}}$. We can choose to let $\tau = \frac{\mu(1-\beta)}{4M}$, such that $\eta_1 - \frac{1}{R}\sum^M_{j=1} \mathbb E[Q^{t-1}_j \Phi^r_j] \leq 5\delta + \frac{6RM}{N^{10}}$,
    %
    % \begin{eqnarray}
    %   \eta_1 - \frac{1}{R}\sum^M_{j=1} \mathbb E[Q^{t-1}_j \Phi^r_j] \leq 5\delta + \frac{6RM}{N^{10}}
    % \end{eqnarray}
    %
    where we add the item $\frac{RM}{N^{10}}$ by taking into the fact that the above inequality condition (i.e., $Q^r_j \geq \frac{\mu(1-\beta)}{4M}$) may not be satisfied in some rounds such that the resulting regret for the corresponding epoch (involving the rounds) is at most $1$. Hence, when $\frac{\ln(1/\tau)}{\delta^2} \leq R \leq \frac{N^{10}\delta}{6M}$, we have $\eta_1 - \frac{1}{R}\sum^M_{j=1} \mathbb E[Q^{t-1}_j \Phi^r_j] \leq 6\delta$. We finally complete the proof of \textbf{Theorem}~\ref{thm:main} by concluding all the conditions on $N$ as shown in the above lemmas.

  \vspace{-2ex}
  \subsection{Privacy Preservation} \label{ssec:prip}
    As shown in Sec.~\ref{ssec:pert}, we design a perturbing mechanism, according to which, each agent can preserve its differential privacy locally when sharing its private knowledge to its untrusted peers in each round. We now prove the efficacy of our proposed perturbing mechanism in \textbf{Theorem}~\ref{thm:rdldp}.
    \begin{theorem} \label{thm:rdldp}
      In each round, our perturbation mechanism $\mathcal F$ achieves $\varepsilon$-LDP for each agent.
    \end{theorem}
    \begin{proof}
      According to our perturbing process (shown in Sec.~\ref{ssec:pert}), for any adoption vectors $\mathbf X_1, \mathbf X_2 \in \{0,1\}^M$ and any perturbed adoption vector $\widetilde{\mathbf X} \in \{0,1\}^M$, we have
      \begin{eqnarray*} \label{eq:ldpsoc1}
        &&{\mathbb P[\mathcal F(\mathbf X_1)=\widetilde{\mathbf X}]}/{\mathbb P[\mathcal F(\mathbf X_2)=\widetilde{\mathbf X}]}\nonumber\\
        &=&\frac{\left(\frac{\exp(\varepsilon/2)}{\exp(\varepsilon/2)+1}\right)^{M-\lVert \widetilde{\mathbf X}-\mathbf X_1\rVert_1}\left(\frac{1}{\exp(\varepsilon/2)+1}\right)^{\lVert \widetilde{\mathbf X}-\mathbf X_1\rVert_1}}{\left(\frac{\exp(\varepsilon/2)}{\exp(\varepsilon/2)+1}\right)^{M-\lVert \widetilde{\mathbf X}-\mathbf X_2\rVert_1}\left(\frac{1}{\exp(\varepsilon/2)+1}\right)^{\lVert \widetilde{\mathbf X}-\mathbf X_2\rVert_1}} \nonumber\\
        %
        %&=&\frac{\exp\left(\frac{\varepsilon}{2}\left(M-\lVert \widetilde{\mathbf X}-\mathbf X_1\rVert_1\right)\right)}{\exp\left(\frac{\varepsilon}{2}\left(M-\lVert \widetilde{\mathbf X}-\mathbf X_2\rVert_1\right)\right)} \nonumber\\
        %
        %&=&\exp\left(\frac{\varepsilon}{2}\left(\lVert \widetilde{\mathbf X}-\mathbf X_2\rVert_1-\lVert \widetilde{\mathbf X}-\mathbf X_1\rVert_1\right)\right) \nonumber\\
        %
        &\le&\exp\left(\frac{\varepsilon}{2}\left(\lVert \mathbf X_1-\mathbf X_2\rVert_1\right)\right) \le \exp(\varepsilon)
      \end{eqnarray*}
      where $\|\cdot\|_1$ denotes $\ell_1$-norm. Then, for $\forall \widetilde{\mathbf X} \subseteq \{0,1\}^M$, 
      \begin{eqnarray*} %\label{eq:ldpsoc2}
        \frac{\mathbb P[\mathcal F(\mathbf X_1)\in\widetilde{\mathcal X}]}{\mathbb P[\mathcal F(\mathbf X_2)\in\widetilde{\mathcal X}]} = \frac{\sum_{\widetilde{\mathbf X}\in\widetilde{\mathcal X}} \mathbb P[\mathcal F(\mathbf X_1)=\widetilde{\mathbf X}]}{\sum_{\widetilde{\mathbf X}\in\widetilde{\mathcal X}} \mathbb P[\mathcal F(\mathbf X_1)=\widetilde{\mathbf X}]} \leq \exp(\varepsilon)
      \end{eqnarray*}
    \end{proof}

    %Base on the above result about per-round privacy loss, we can directly apply the composition rule of LDP to conclude that the total privacy loss over $R$ rounds is at most $R\varepsilon \leq \frac{N^{10}\delta}{6M}\varepsilon$.

    %Furthermore, given the privacy level $\varepsilon$ the selection of $N$ may depend on the selection of $g(N)$. If we select a higher-order function $g(N)$ such as $\sqrt{N}$, then we might achieve $\varepsilon$-privacy by taking a relative small $N$ than selecting lower-order $g(N)$. However, a higher-order $g(N)$ would in one hand cause more consumption of communication which increases the load of the network system, and on the other hand cost more time for message dissemination in each round. A lower-order $g(N)$(e.g. $\ln^2 N$) would be the choice if we aim to reduce the network burden and accelerate the learning process when we have more nodes/agents to participate. Hence, the selection of $g(N)$ should be a trade-off between the consumption of communication/time and the computing resources according to the actual circumstance.

\vspace{-4ex}
\section{Simulations} \label{sec:exp}
  As mentioned in \textbf{Remark}~\ref{re:main}, \textbf{Theorem}~\ref{thm:main} actually implies the impacts of the number of agents $N$, the number of unknown options $M$ and privacy budget $\varepsilon$ on the regret (or learning utility) of our algorithm. Therefore, in this section, we perform extensive numerical simulations to empirically reveal the impacts of the above different parameters in addition to the theoretical analysis. In the following simulations, social graph $\mathcal G= (\mathcal N, \mathcal E)$ is constructed in a randomized manner. Specifically, given a group of agents $\mathcal N$, we randomly add edges such that the resulting graph is connected and non-bipartite. According to \textbf{Theorem}~\ref{thm:main}, we fix constants $\beta=0.505$, $\sigma=15$, $\delta=0.02$, $\mu=6.7\times 10^{-5}$ and $h=485$, as these constants actually have much less impact on the empirical analysis on our algorithm. Note that all our empirical analysis still holds when the constant parameters take another values. For each reported data sample, we repeat the experiments for thirty times and take average over the results.

  \subsection{Learn More if Paying More}\label{ssec:simn}
    We hereby first investigate the convergence of the regret of our algorithm with different numbers of agents. According to \textbf{Definition}~\ref{def:funcg}, we choose the following two definitions of $g(N)$, i.e., $g(N)=\ln^2 N$ and $g(N)=\sqrt N$, respectively. We also gradually increase the number of agent such that $N = 3, 6, 10 \times 10^3$. To concentrate on revealing the impacts of $N$ on the convergence of our algorithm, we fix $M=10$ and let $\eta_1, \cdots, \eta_M$ be uniformly distributed in the range of $[0,1]$. For the same reason, we also fix $\varepsilon=1.0$ to guarantee the local privacy for each of the agents. The experiment results are illustrated in Fig.~\ref{fig:diffn}. As shown in Fig.~\ref{fig:logn} where $g(N)=\ln^2 N$, when $N$ is smaller (e.g., $N=3 \times 10^3$), although the regret of our algorithm converges to a stable level, we cannot guarantee that it can be upper-bounded by $6\delta$. Furthermore, consistent with \textbf{Theorem}~\ref{thm:main}, if we gradually increase $N$ (e.g., let $N=6, 10 \times 10^3$) such that $N$ is sufficiently large with respect to $M$, $\varepsilon$ and $\delta$, an upper bound of $6\delta$ on the regret of our algorithm can be ensured, when our algorithm achieves convergence. In particular, we have smaller regret if letting more agents participate in the social learning process (as explained in \textbf{Remark}~\ref{re:main}). Another observation is that our algorithm achieves convergence within almost the same time horizon, even more agents are engaged. This is not surprising, since the lower bound on the time horizon for our algorithm to converge mainly depends on the number of options $M$, while $M$ is fixed in our case. According to \textbf{Theorem}~\ref{thm:mhrw}, enabling collaboration among an increasing number of agents implies higher communication overhead, but this is the price for higher learning utility and thus smaller regret. When defining $g(N)=\sqrt N$ in Fig.~\ref{fig:sqrtn}, we get very similar results. From the above observations, we learn the following lesson: \textit{by letting more agents participating in the social learning process, our algorithm results in smaller regret and higher learning utility through fully exploiting their collaboration, while ensuring the LDP for each agent.}
    \begin{figure}
      \centering
      \subfigure[$g(N)=\ln^2 N$]{
        \begin{minipage}[b]{0.47\linewidth}
          \includegraphics[width=\linewidth]{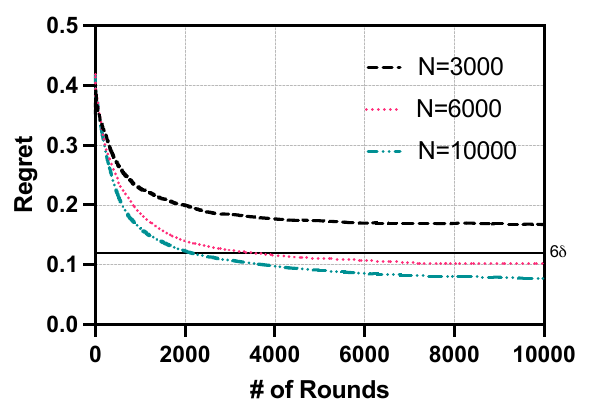}
        \end{minipage}
        \label{fig:logn}}
      \subfigure[$g(N)=\sqrt{N}$]{
        \begin{minipage}[b]{0.47\linewidth}
          \includegraphics[width=\linewidth]{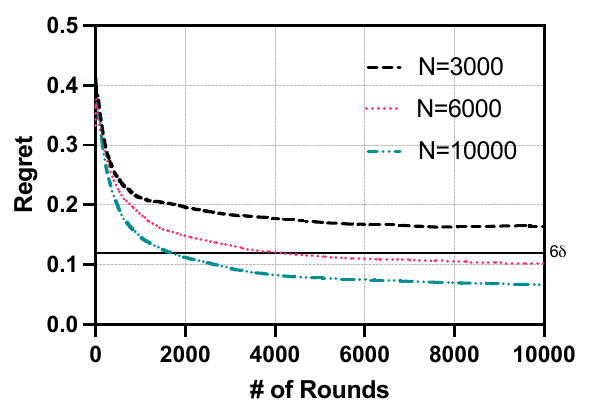}
        \end{minipage}
        \label{fig:sqrtn}}
      \caption{Convergence with different numbers of agents and different choices of $g(N)$ in the first $10^4$ rounds. We fix $M=20$ and $\varepsilon=1.0$.}
      \label{fig:diffn}
    \vspace{-3ex}
    \end{figure}

  %\vspace{-2ex}
  \subsection{What If More Unknown Options Are Given?} \label{ssec:simm}
    We hereby evaluate the performance of our algorithm with different numbers of unknown options. We gradually increase $M$ from $10$ to $30$ with a step size of $10$. We also vary the number of agents such that $N=3, 6, 10 \times 10^3$ and fix $\varepsilon=1.0$. Since we get similar results with $g(N)=\ln^2 N$ and $g(N)=\sqrt N$, we only report the ones with $g(N) = \ln^2 N$, especially considering the limited space. As shown in Fig.~\ref{fig:diffm}, when there are more unknown options for the agents to learn, the regret of our algorithm still converges but to a larger value. Specifically, when there are not a sufficient number of agents while the number of unknown options is too large (e.g., $M \geq 20, N=3\times 10^3$ or $M=30, N=6\times 10^3$), the regret even cannot be bounded by $6\delta$, when our algorithm achieves convergence. Nevertheless, when we increase $N$ to $10^4$, the upper bound holds for $\forall M \leq 30$, by fully exploiting the collaboration among the large number of agents. It is revealed by the above observations that, \textit{even when there are too many unknown options for a given number of agents to learn, our algorithm achieves convergence with a loss in learning utility, while the wisdom we learnt in Sec.~\ref{ssec:simn} suggests us to let more agents participate in the social learning process to obtain higher learning utility and thus smaller regret in face of a large number of unknown options.}
    \begin{figure*}
      \centering
      \subfigure[$N=3 \times 10^3$]{
        \begin{minipage}[b]{0.28\linewidth}
          \includegraphics[width=.99\linewidth]{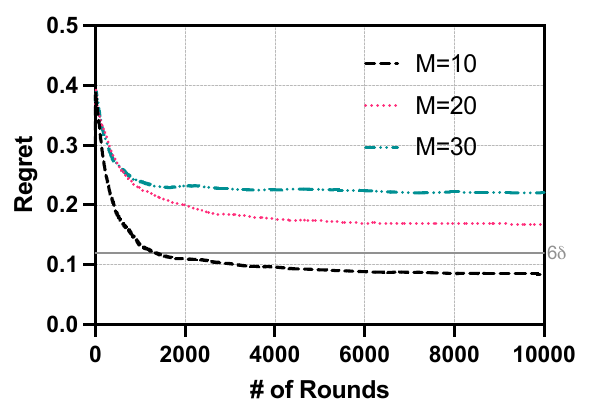}
        \end{minipage}
        \label{fig:logm3k}}
      \subfigure[$N=6 \times 10^3$]{
        \begin{minipage}[b]{0.28\linewidth}
          \includegraphics[width=.99\linewidth]{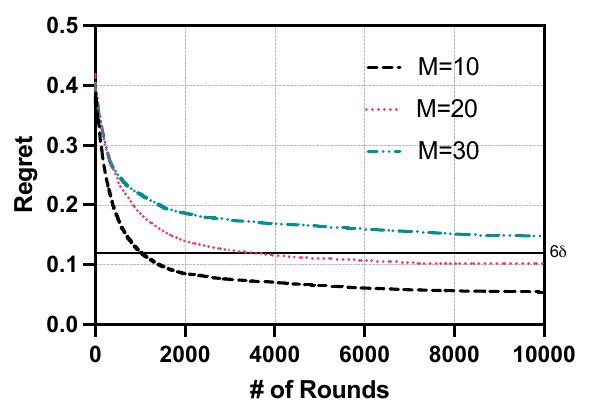}
        \end{minipage}
        \label{fig:logm6k}}
      \subfigure[$N=10 \times 10^3$]{
        \begin{minipage}[b]{0.28\linewidth}
          \includegraphics[width=.99\linewidth]{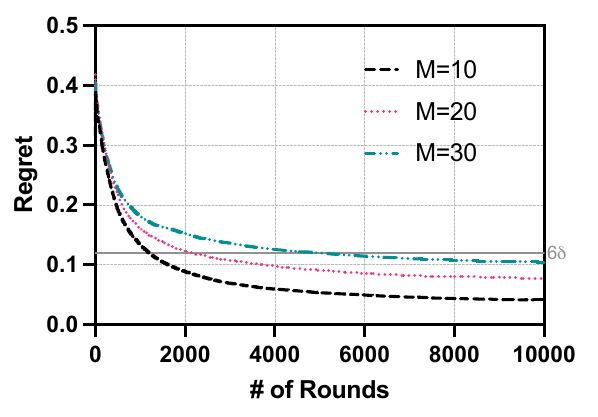}
        \end{minipage}
        \label{fig:logm10k}}
      \vspace{-2ex}
      \caption{Convergence of the regret function with different numbers of options. We let $g(N)=\ln^2 N$ and fix and $\varepsilon=1.0$.}
      \label{fig:diffm}
    \vspace{-2ex}
    \end{figure*}

  \vspace{-2ex}  
  \subsection{Privacy vs. Utility} \label{ssec:exppri}
    As mentioned in Sec.~\ref{sssec:ldp}, we can tune the privacy budget $\varepsilon$ to adjust the ``strength'' of the privacy preserving. When $\varepsilon$ is small, less privacy loss is allowed such that each agent has to introduce more perturbation on their adoption vectors, while more perturbation usually implies less learning utility and thus higher regret for our social learning algorithm. As suggested by \cite{ShokriS-CCS15,AbadiCGMMTZ-CCS16,GajaneUK-ALT18,WangZWCKW-RecSys20,RenZLS-arXiv20}, we gradually decrease $\varepsilon$ from $2.0$ to $0.5$ with a step size of $0.5$, and illustrate the impact of $\varepsilon$ on the regret of our algorithm in different networks with $N=3, 6, 10 \times 10^3$. Likewise, we choose $g(N)=\ln^2 N$ and fix $M=20$ to facilitate our empirical analysis on privacy preserving. The simulation results are reported in Fig.~\ref{fig:privacy}. We also plot the results with $\varepsilon=\infty$ as comparison, where each agent does not perturb its adoption vector such that no privacy is guaranteed. As demonstrated in  Fig.~\ref{fig:privacy}, given a group of agents participating in the social learning process, we indeed obtain higher regret and thus less learning utility by decreasing the privacy budget $\varepsilon$ and introducing more perturbation. Considering smaller privacy budget results in stronger privacy preserving as shown in \textbf{Theorem}~\ref{thm:rdldp}, this is the price we have to pay. Similar to what we have observed in Fig.~\ref{fig:diffn}, by increasing the number of agents, we have the regret converge to a larger value, while guaranteeing the privacy preserving at a high level. For example, when $N=3 \times 10^3$, the converged regret even cannot be bounded by $6\delta$ with $\varepsilon \leq 1.5$. When $N$ is increased to $10^4$, it is decreased significantly such that the upper bound $6\delta$ holds for $\forall \varepsilon \in [0.5, 2.0]$. Especially, the regret with $1.0 \leq \varepsilon \leq 2.0$ is very close to the one with $\varepsilon=\infty$. That is, when there are sufficient agents, the sacrifice in learning utility for privacy could be very little. All in all, the main conclusion we get from these simulation results is, \textit{although higher demand on privacy preserving results in a sacrifice in learning utility and this is the price we have to pay, we are able to manipulate the trade-off between privacy preserving and learning utility by leveraging the number of agents participating in the social learning process.}
    \begin{figure*}
      \centering
      \subfigure[$N=3 \times 10^3$]{
        \begin{minipage}[b]{0.28\linewidth}
          \includegraphics[width=.99\linewidth]{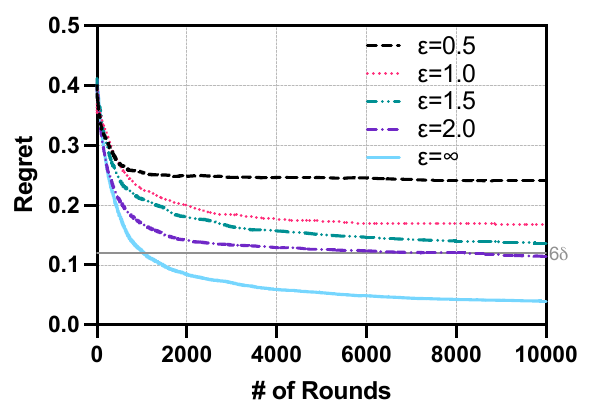}
        \end{minipage}
        \label{fig:privacy3k}}
      \subfigure[$N=6 \times 10^3$]{
        \begin{minipage}[b]{0.28\linewidth}
          \includegraphics[width=.99\linewidth]{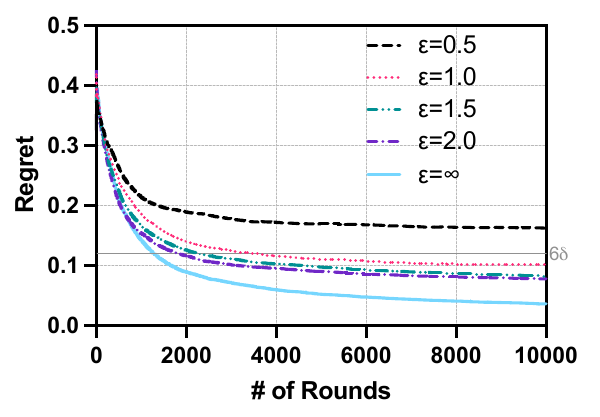}
        \end{minipage}
        \label{fig:privacy6k}}
      \subfigure[$N=10 \times 10^3$]{
        \begin{minipage}[b]{0.28\linewidth}
          \includegraphics[width=.99\linewidth]{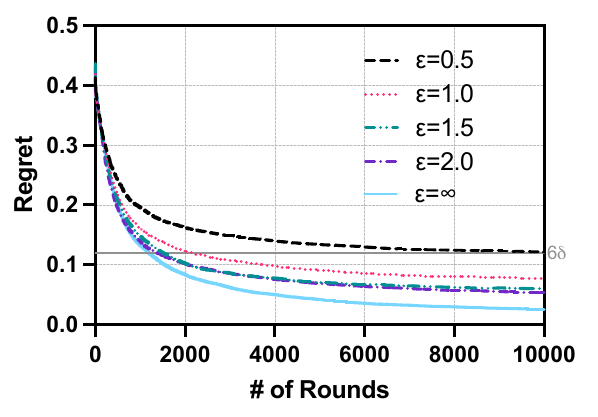}
        \end{minipage}
        \label{fig:privacy10k}}
      \vspace{-2ex}
      \caption{Convergence of the regret function with different privacy budgets and different numbers of agents. We let $g(N)=\ln^2 N$ and fix $M=20$.}
      \label{fig:privacy}
    %\vspace{-2ex}
    \end{figure*}

%\vspace{-5ex}
\section{Conclusion and Future Work} \label{sec:con}
  In this paper, we have presented a distributed privacy-preserving social learning algorithm for general social networks. We leverage the notion of LDP such that each agent in the social network perturbs its private adoption for privacy preserving. We also utilize random walks to realize efficient experience sharing among the agents over the social network with general topology. We have performed solid theoretical analysis to show that when there are a sufficiently number of agents participating in the social learning process, the regret of our algorithm is bounded by a constant with affordable communication overhead (see \textbf{Theorem}~\ref{thm:mhrw} and \textbf{Theorem}~\ref{thm:main}), while the differential privacy of the agents can be preserved locally (see \textbf{Theorem}~\ref{thm:rdldp}). Extensive simulations also have been performed to empirically study the trade-off among the number of agents (or communication overhead), privacy preserving and learning utility.

  As shown in \textbf{Theorem}~\ref{thm:main}, we have quantified the trade-off between the privacy and the utility of our proposed social learning algorithm. Another interesting problem is, what is the minimum amount of noise (or perturbation) added to achieve the highest utility while preserving the differential privacy? The problem has been investigated in many recent proposals~\cite{BalleW-ICML18,GengDGK-AISTATS20}; nevertheless, these proposals characterize the optimal trade-off between the privacy and utility for the global DP model, while it is highly non-trivial to derive the minimum amount of noise under the LDP model in our decentralized social learning process.

  Another possible research direction for future is to consider asynchronous multi-agent systems. In this paper, we assume the agents are synchronized, while such an assumption may not always be available. Especially, for a large-scale multi-agent system, it is very difficult to synchronize the agents, while how to exploit efficient collaboration among the asynchronous agents is significantly challenging.

\bibliographystyle{IEEEtran}
\bibliography{ppsoc}

\end{document}